\documentclass[12pt,reqno]{amsart}
\DeclareMathOperator{\E}{E}

\usepackage[foot]{amsaddr}
\usepackage{booktabs}

\usepackage{amsthm}
\usepackage{tabularx}
 
\usepackage{mathtools}

\usepackage{amssymb}
\usepackage{amsmath}
\usepackage{amsfonts}
\usepackage{mathrsfs}
\usepackage{nicefrac}
\usepackage{graphicx}
\usepackage{threeparttable}
\usepackage[referable]{threeparttablex}
\usepackage{subcaption}
\usepackage{bmpsize}
\usepackage[onehalfspacing]{setspace}
\usepackage{caption}
\usepackage[shortlabels]{enumitem}
\usepackage[utf8]{inputenc}
\usepackage{charter}
\usepackage[colorlinks=true,linkcolor=blue,citecolor=blue,urlcolor=blue,pdfpagemode=UseNone,pdfstartview=FitH]{hyperref}
\usepackage{apptools}
\usepackage{color}
\usepackage{natbib}

\usepackage{bbm}

\usepackage{xcolor}

\makeatletter
\def\section{\@startsection{section}{1}
	\z@{1.0\linespacing\@plus\linespacing}{.8\linespacing}{\Large}}

\def\subsection{\@startsection{subsection}{2}
	\z@{.8\linespacing\@plus.7\linespacing}{.7\linespacing}{\large}}

\def\subsubsection{\@startsection{subsubsection}{3}
	\z@{.5\linespacing\@plus.7\linespacing}{-.5em}{\normalfont\bfseries}}
\makeatother

\numberwithin{equation}{section}

\newtheorem{proposition}{Proposition}[section]

\theoremstyle{definition}

\theoremstyle{definition}
\newtheorem{assumption}{Assumption}[section]

\theoremstyle{definition}

\theoremstyle{definition}

\allowdisplaybreaks
\usepackage{setspace}

\usepackage{multirow}

\expandafter\let\expandafter\oldproof\csname\string\proof\endcsname
\let\oldendproof\endproof
\renewenvironment{proof}[1][\proofname]{%
  \oldproof[\normalfont \bfseries  #1]%
}{\oldendproof}

\title[Managing Procurement Auction Failure]{Managing Procurement Auction Failure:\\ Bid Requirements or Reserve Prices?}

\thanks{We thank our editor, Jean-Francois Houde, four anonymous referees, Xiaohong Chen, Li Hao, Tong Li, Michael Peters, Wing Suen, Balazs Szentes, and Janne Tukiainen for their helpful comments and discussions. We acknowledge the financial support from the National Natural Science Foundation of China under grant 72394392 (Ma), the Social Sciences and Humanities Research Council of Canada under grants 435-2021-0189 and 435-2025-0755 (Marmer), the Natural Sciences and Engineering Research Council of Canada under grant RGPIN-2024-04877 (Marmer), HK Grant Council under General Research Funding under grant 17503720 (Xu), and HKU Business School under its visiting scholar program. Shiming Wu provided excellent research assistance. Refine.ink was used to check the paper for consistency and clarity.}
\author{Jun Ma\textsuperscript{*}}
\address{\textsuperscript{*}School of Economics, Renmin University of China, jun.ma@ruc.edu.cn}
\author{Vadim Marmer\textsuperscript{\textdagger}}
\address{\textsuperscript{\textdagger}Vancouver School of Economics, University of British Columbia, vadim.marmer@ubc.ca}
\author{Pai Xu\textsuperscript{\S}}
 \address{\textsuperscript{\S}Business School, University of Hong Kong, paixu@hku.hk}

\begin{document}

\begin{abstract}


This paper examines bid requirements, where the government may cancel a procurement contract unless two or more bids are received. Using a first-price auction model with endogenous entry, we compare the bid requirement and reserve price mechanisms in terms of auction failure and procurement costs. 
We find that, in comparison with bid requirements, reserve prices can reduce procurement costs and substantially lower failure probabilities, especially when entry costs are high or signals are sufficiently informative. Bid requirements are more likely to result in zero entry, while reserve prices can sustain positive entry under broader conditions.
\vspace{1ex}

{\footnotesize \noindent \textsc{Keywords}: First-price auctions; procurement auctions; reserve price; endogenous entry; informative signals; semiparametric estimation}

{\footnotesize \noindent \textsc{JEL Classification}: C12; C13; C14}
\end{abstract}

\date{\today} 
\maketitle



\section{Introduction}\label{sec:intro}

In procurement auctions, the government can cancel a contract if not enough bids are received. Such ``bid requirements'' ensure the reasonableness of the auction price through adequate competition. Many countries have adopted rules and regulations with bid requirements: for instance, the Czech Republic requires at least two bidders \citep{titl2023one}, China mandates a minimum of three qualified bids, and various U.S.\ federal and state regulations impose similar conditions (see Section~\ref{sec:institutional} for a detailed discussion). The rationale behind bid requirements is that a lack of competition may result in high procurement costs. Unfortunately, bid requirements may also lead to auction failure if not enough bidders choose to participate, generating significant costs for the government and society.

Alternatively, procurement costs can be controlled with reserve prices, which are effective even when competition is low. However, in practice, the government may lack a clear benchmark for setting the reserve price.
Moreover, an aggressively set reserve price can itself increase the probability of auction failure if no submitted bids fall below it. Both formats address the same fundamental concern, establishing the reasonableness of the auction price, but through different mechanisms: bid requirements do so by ensuring adequate competition, while reserve prices provide a direct price benchmark.

This paper compares the bid requirement and reserve price formats by examining their probabilities of auction failure and expected procurement costs. Our main finding is that the reserve price format can offer considerable advantages over the bid requirement format, especially in preventing auction failure. Even with relatively aggressive reserve prices, the reserve price format may result in substantially smaller probabilities of auction failure while maintaining lower procurement costs.

Our analysis utilizes the endogenous entry model \citep{ye2007indicative, MSX, GL}. In this model, potential bidders receive imperfect but informative signals about their valuations (private costs of completing the contract). They must incur an entry cost to learn their true valuations and enter. In equilibrium, bidders enter if their signals are below an endogenously determined cutoff. Only entering bidders participate in the bidding stage, where they submit sealed bids without knowing the number of rivals. Under the bid requirement format, the contract is awarded to the lowest bidder provided that at least two bids are submitted; under the reserve price format, an award is made if at least one bid falls below the reserve price. Failing these conditions, the auction is canceled.

The equilibrium entry cutoff is determined by a zero expected profit condition for the marginal entrant, that is, the bidder whose signal equals the entry cutoff. In equilibrium, the marginal entrant is, therefore, indifferent between entering and not entering.

A key disadvantage of the bid requirement format is that entry breaks down at lower entry costs than under reserve prices. This arises because the bid requirement format creates non-monotone incentives: as the entry probability increases from zero, an active bidder's expected revenue initially rises because the auction is less likely to be canceled, but eventually falls as competition intensifies. Under the reserve price format, by contrast, expected revenue is monotonically decreasing in the entry probability and is highest when the entry probability is near zero, since a sole bidder can still win. Under the bid requirement format, the auction will likely be canceled in such low-entry situations, resulting in zero revenue. Consequently, the reserve price format generates higher expected revenues at low entry probabilities and can sustain entry at higher entry costs.

One of our starkest theoretical predictions is that entry stops completely under the bid requirement format when signals are sufficiently informative, guaranteeing auction failure. The intuition can be understood by focusing on the limiting case where signals are perfectly informative, as in \cite{samuelson1985competitive}. In this case, entry unravels because the marginal entrant has a zero probability of winning: their valuation is precisely on the entry threshold, so all other active bidders are guaranteed to have better valuations. The marginal entrant can only win if there are no other active bidders, but then the auction is canceled under the bid requirement format.\footnote{We thank an anonymous referee for providing this illustration via the \cite{samuelson1985competitive} limiting case.} This mechanism continues to disadvantage the marginal entrant when signals are imperfect but sufficiently informative. The reserve price format, by contrast, can support positive entry for all levels of signal informativeness.

Toward the other end of the spectrum, when signal informativeness is low, the marginal entrant may have a valuation and probability of winning similar to those of other active bidders. As a result, the bid requirement format can lead to lower probabilities of auction failure than the reserve price format with an aggressively set reserve price. Thus, the relative ranking of the two formats is determined by the level of signal informativeness, and characterizing this dependence is a distinctive contribution of our analysis.

To evaluate the empirical relevance of these theoretical mechanisms, we revisit the application from \cite{li2009entry}, which examines highway maintenance procurement auctions conducted by the Texas Department of Transportation (TxDoT) between 2001 and 2003. As we discuss in detail in Section~\ref{sec:institutional}, the institutional environment of TxDoT auctions, in particular, the non-enforcement of the engineer's estimate as a reserve price and the near-absence of single-bidder auctions in the data, motivates modeling these auctions under the bid requirement format.

According to our estimates, the expected procurement cost is $2.0$--$3.7$ percentage points (of the engineer's estimate) lower under the reserve price format than under bid requirements. Central to our findings are the striking differences in auction failure probabilities. For auctions with $9$--$12$ potential bidders, the probability of auction failure is estimated at $23\%$--$26\%$ under the reserve price format and $44\%$--$51\%$ under the bid requirement format. These differences arise because estimated entry probabilities are relatively low, ranging from $13\%$ to $19\%$, leading to a high likelihood of only one active bidder and subsequent project cancellations under the bid requirement format.

Our results also suggest an explanation for why the reserve price was not enforced in practice: the engineer's estimate corresponds to only the $29$th percentile of the valuations distribution. Enforcing a binding reserve price at this level would exclude $71\%$ of potential realizations, creating a significant trade-off between price control and participation.

We find that signal informativeness is moderately high, with Spearman's rank correlation of $0.68$ between valuations and signals. However, it is plausible that the level of signal informativeness would be different in other applications. We therefore use the estimated structural model to investigate the relative performance of the bid requirement and reserve price formats under different levels of signal informativeness.

Theoretically, the effect of signal informativeness is ambiguous, as changing its level triggers two competing mechanisms: an ``information'' effect and a ``cutoff'' effect. The information effect reflects the direct impact of signal informativeness: holding the entry probability constant, more informative signals reduce expected procurement costs because entering bidders tend to have lower private costs. In contrast, the cutoff effect operates indirectly through the equilibrium entry probability, which may increase or decrease in response to more precise information. Consequently, the overall impact of signal informativeness on the expected procurement cost remains ambiguous.

We find that under the bid requirement format and for all entry probabilities, the marginal entrant's revenue decreases with signal informativeness. Given the previously discussed non-monotone incentives, this generates discontinuous entry patterns under bid requirements. Specifically, entry ceases completely when signals are sufficiently informative, confirming our theoretical prediction. For example, in the case of $10$ potential bidders, according to our estimates, auctions under the bid requirement format fail with probability one if Spearman's rank correlation between private costs and signals exceeds $0.72$. The reserve price format, on the other hand, can support the entire range of signal informativeness levels with strictly positive entry probabilities. Thus, we conclude that the reserve price format is particularly advantageous not only when entry costs are high but also when signals are sufficiently informative.

For robustness, we also estimate a ``soft'' bid requirement model in which, following \cite{li2009entry}, the government enters as a bidder when only one bid is submitted, rather than canceling the auction outright. This alternative specification produces higher estimates of the entry cost but yields very similar estimates of the other model fundamentals: the informativeness of signals and the distribution of private costs. Our main conclusions regarding the relative strengths of the two formats remain unchanged.

Because signals are unobserved and there are no exogenous entry-cost shifters, the model is not fully identified nonparametrically. To overcome these challenges, we employ a semiparametric approach initially discussed in \cite{GL}, modeling the joint distribution of valuations and signals through a parametric copula function with a single parameter while treating the marginal distribution of valuations nonparametrically. The copula parameter provides a direct measure of signal informativeness through its one-to-one correspondence with Spearman's rank correlation.

Methodologically, we extend this framework by establishing formal identification conditions and developing a practical GMM-type estimation procedure that utilizes the bootstrap to compute the efficient weight matrix. Some of the details of our econometric procedure are of independent interest. For example, in the online Supplement to this paper, we derive the asymptotic distribution of the empirical cumulative distribution function (CDF) of the estimated private costs. Combined with the derivations in \cite{Ma2019}, the results can be used for inference on the CDF and the probability density function (PDF) of latent private costs.

Our paper contributes to the literature on entry in auctions. Earlier papers include \cite{samuelson1985competitive} and \cite{levin1994equilibrium}, who focus on the issues of costly entry and bidders' uncertainty about their values, respectively. In a more recent paper, \cite{li2009entry} discuss the selection between different entry models and study the impact of the number of potential bidders on expected procurement costs. This issue is also examined in \cite{MSX}, who develop formal tests for \cite{samuelson1985competitive}, \cite{levin1994equilibrium}, and the endogenous entry model with affiliated signals. They also discuss the issue of selective entry in the endogenous entry model. The endogenous entry model with affiliated signals is also the subject of \cite{GL}, who study its nonparametric identification, and \cite{chen2024identification}, who study the identification of the model under risk aversion. The literature has also considered several important issues such as indicative bidding \citep{roberts2013should}, estimation of entry costs \citep{xu2013nonparametric}, bidders' risk attitudes \citep{fang2014inference}, entry rights \citep{bhattacharya2014regulating}, and sequential bidding \citep{quint2018theory}. Our paper also contributes to an extensive literature on competition and procurement costs \citep[see, e.g.,][and many others]{hong2002increasing, li2009entry, krasnokutskaya2011bid}.

The issues of low competition and the prevalence of single-bid auctions in procurement have been discussed in the literature \citep{kang2022winning,titl2023one}. However, the topic of auction failure due to insufficient participation is largely overlooked. We contribute to the literature by investigating mitigating mechanisms.



The rest of the paper proceeds as follows. Section~\ref{sec:institutional} provides institutional background and motivates our modeling choices. Section \ref{sec:model} presents the model and analyzes the two formats. In Section \ref{sec:Signals_cost}, we investigate the role of signal informativeness in the comparison of the two formats. Section \ref{sec:ID} discusses the semiparametric identification and estimation of the model under bid requirements. Section \ref{sec:empirical} presents our estimation results. In Section \ref{sec:counterfactuals}, we compare the two formats in terms of the failure probability and expected procurement cost at the estimated and varying levels of signal informativeness. Section \ref{sec:soft_bid_req} gives results on the equilibrium, identification, and estimation of the model with the soft bid requirement and presents estimation results under the assumption of a soft bid requirement rule. The online Supplement contains additional theoretical and empirical results, econometric details, and their proofs.\footnote{The Supplement is available at \href{https://ruc-econ.github.io/auction-supp-v12.pdf}{https://ruc-econ.github.io/auction-supp-v12.pdf}.}


\section{Background}\label{sec:institutional}


Bid requirements (i.e., rules that cancel an auction if too few bids are received) are widely used in public procurement. The 2012 public procurement reform in the Czech Republic introduced the requirement of at least two bidders, making it illegal to award a contract if only one bid was received \citep{titl2023one}. China enacted a law in 2000 stipulating that a procurement process is deemed unsuccessful unless a minimum of three qualified bids are received.

In the U.S., the Federal Acquisition Regulation (Section 14.404) states that the agency may decide to cancel the auction if only one bid is received and the contracting officer cannot determine the reasonableness of the price. State-level laws may impose more stringent rules. For example, North Carolina (General Statutes 143-132) stipulates that for certain projects, ``no contract shall be awarded unless at least three competitive bids have been received''. These rules reflect a broader principle in U.S.\ federal procurement guidelines that at least two bids are needed for adequate price competition; see, e.g., the Federal Transit Administration's best practices \citep[][p.~101]{FTA2016_Best}.

When an auction fails due to insufficient bids, the consequences extend beyond the administrative costs of relisting the project. In the case of highway maintenance procurement, auction failure may lead to delayed maintenance, which can result in reduced traffic safety and increased vehicle operating costs. These social costs can be substantial, particularly for time-sensitive infrastructure projects.

Our empirical application examines highway maintenance procurement auctions conducted by the Texas Department of Transportation (TxDoT) between 2001 and 2003, using the data set from \cite{li2009entry}. For each auction, the data include the number of documentation requests received by TxDoT, providing information on the number of potential bidders.

According to TxDoT rules, the contract is awarded to the bidder with the lowest bid, with notable exceptions. Under Title 43 of the Texas Administrative Code, the TxDoT commission will cancel the auction if ``the lowest bid is higher than the department's estimate and the commission determines that re-advertising the project for bids may result in a significantly lower low bid.'' More generally, TxDoT may reject a contract when it is ``in the best interest of the State'' \citep{TxDOT2014}.

However, as discussed in \cite{li2009entry}, the engineer's estimate is not enforced as a reserve price in practice, with many winning bids exceeding the published estimates. For example, in auctions with $9$--$12$ potential bidders, between $13\%$ and $44\%$ of winning bids are above the engineer's estimate. Using a similar argument, \cite{krasnokutskaya2011identification} justifies treating the reserve price as non-enforced in Michigan highway procurement data, where many bids exceed the engineer's estimate. There may be good reasons for this non-enforcement: the administration requires the estimate to be finalized months before the bid opening date, so it may not reliably reflect current construction or maintenance costs.

The non-enforcement of the reserve price, combined with a striking pattern in the data, motivates modeling TxDoT auctions under the bid requirement format. As reported in \cite{li2009entry}, only a negligible fraction of auctions have just one active bidder (13 auctions, or $2.35\%$ of the entire sample). Yet, given the observed entry probabilities, one would expect a substantially higher fraction: approximately $28\%$--$30\%$ for cases with $9$--$12$ potential bidders. The large number of ``missing'' single-bidder auctions strongly suggests that official procurement guidelines and prevailing industry practices have shaped the expectations of both auctioneers and bidders that auctions with insufficient competition will be canceled. Our main analysis therefore adopts the bid requirement format as the data-generating process for the TxDoT auctions.

While the bid requirement format serves as a natural benchmark, mapping it cleanly to practice presents challenges. \cite{li2009entry} report a small number of auctions with only one active bidder that were nevertheless awarded. Although these auctions are excluded from their TxDoT data set, we also estimate a ``soft'' bid requirement model in which, following the assumption in \cite{li2009entry}, the government enters as a bidder when only one bid is submitted, rather than canceling the auction outright. This alternative specification provides a robustness check. It produces higher estimates of the entry cost but yields very similar estimates of the other model primitives, and our main conclusions remain unchanged.

\section{Model and auction formats}\label{sec:model}

In this section, we describe the auction model with endogenous entry and discuss
how strategies and expected outcomes change under two alternative formats: (i)
``bid requirement,'' where at least two active bidders must participate in bidding for the
contract to be awarded but no reserve price is set, and (ii) ``reserve price,'' where the contract
is awarded if at least one of the submitted bids is below the reserve price.

\subsection{Model}

We follow \citet{MSX} while adapting the model to low-bid
procurement auctions. 
Let $n\geq 2$ denote the number of potential bidders. 
At the entry stage, each potential bidder draws a private cost $V$ and a signal $T$. The draws are from the same joint distribution for all potential bidders and are independent across bidders. At that stage, a bidder does not know their private cost  $V$ of completing the contract but can learn it after paying the entry cost $\kappa$. After paying the entry cost, a potential bidder becomes an active bidder and proceeds to the bidding stage, knowing their $V$. The decision to enter is based on the drawn signal $T$, the entry cost $\kappa$, the joint distribution of $V$ and $T$, and the number of potential bidders. The joint distribution and the number of potential bidders are known to all participants. Potential bidders are risk-neutral, and their project completion costs and signals are private.

Only active bidders who have paid the entry cost $\kappa$ may participate in the bidding. The number of active bidders is unknown to the participants. Sealed bids are submitted, and the contract is awarded to the lowest bidder.

Let $F_{V,T}(v,t)$  denote the joint CDF of private contract completion costs $V$ and private signals $T$. It is convenient to express this CDF in terms of the copula function: $F_{V,T}(v,t)=C(F(v),F_T(t))$, where $C(\cdot,\cdot)$ is the copula function, and $F(\cdot)$  and $F_T(\cdot)$ are the marginal CDFs of $V$ and $T$, respectively. While the private costs $V$ are unobserved, they are identified conditional on entry and can be estimated using bid data. On the other hand, signals $T$ are unobserved and unidentified. It is convenient to use the ranks of signals instead. Define $S\coloneqq F_T(T)$, and write the joint CDF of costs $V$ and $S$ as $F_{V,S}(v,s)\coloneqq C(F(v),s)$. For simplicity, we refer to $S$ as signals in what follows. We make the following assumption.

\begin{assumption}\label{a:copula_1}
\begin{enumerate}[(i)]
    \item[]
    
    \item The copula function $C(\cdot,\cdot)$ is thrice continuously differentiable. \label{a:C_diffble}
    \item $C_{22}(x,y)\coloneqq \partial^2C(x,y)/\partial y^2< 0$ for all $x,y\in (0,1)$.\label{a:goodnews}
    \item The CDF $F(\cdot)$ of the private costs $V$ of completing the contract is absolutely continuous with a continuous PDF denoted by $f(\cdot)$. 
    \item $f(\cdot)$ is bounded away from zero on its compact support $[\underline{v},\overline{v}]$.\label{a:marginal_cdf}
\end{enumerate}
\end{assumption}

By the copula properties, the conditional distribution of the private cost given the signal $S$ can be expressed as $F_{V\mid S}(v\mid s)=C_2(F(v),s)$, where $C_2(x,y)\coloneqq \partial C(x,y)/\partial y$.  Therefore, Assumption \ref{a:copula_1}\ref{a:goodnews} is equivalent to a first-order stochastic dominance relationship for the conditional distribution of private costs given signals: for all $v\in (\underline v,\overline v)$ and  $0\leq s_1<s_2\leq 1$,
\[
F_{V\mid S}(v \mid s_1) > F_{V\mid S}(v \mid s_2).
\]
In \cite{MSX}, this condition was called the ``good news'' assumption:
drawing a smaller signal corresponds to having a stochastically smaller contract completion cost. 

We consider only pure-strategy symmetric equilibria. Under the ``good news'' assumption, the entry strategy is to enter if a sufficiently small signal is drawn. That is, a bidder with a signal $S$ enters when $S\leq p$, where the entry cutoff $ p\in[0,1]$ is determined in equilibrium. This is because a bidder's expected profit from entry, conditional on their signal $S=s$, is a decreasing function of $s$ under the ``good news'' assumption.\footnote{See the discussion following Proposition \ref{p:entry} below.} Since $S$ is uniformly distributed by construction, $p$ is also the probability of entry.

 Below, we discuss the equilibrium entry and bidding strategies, expected winning bids, and failure probabilities under the two auction formats.

\subsection{Bid requirement format}\label{sec:bid_requirement}

Under the bid requirement format, the contract is awarded only if two or more bids are submitted; however, there is no reserve price. An auction fails when there are no active bidders or only one active bidder.

Given an entry cutoff (or probability) $p\in(0,1)$, the distribution of $V$ conditional on entry is given by
\begin{align}\label{eq:F^*}
    F^*(v\mid p)&\coloneqq \Pr[V\leq v\mid S\leq p] ={C(F(v),p)}/{p}.
\end{align}
Let $H(v\mid p,n)$ denote the probability of an active bidder with a private cost $v$ winning the contract when the entry probability of the competitors is $p$:
\begin{align*}
    H(v\mid p,n) &\coloneqq\Lambda^{n-1} (v\mid p)-(1-p)^{n-1},\,\text{where}\\
    \Lambda(v\mid p)&\coloneqq 1-p +p\cdot (1-F^*(v\mid p)) 
    = 1-C(F(v),p).
\end{align*}
The $\Lambda(\cdot\mid p)$ component of $H(\cdot\mid p,n)$ captures the probability that a competitor either does not enter or draws a private cost above $v$. The $(1-p)^{n-1}$ component reflects the bid requirement for at least one other active bidder. 
 
Following standard arguments, the equilibrium bidding strategy is given by
\begin{equation}
    \beta(v\mid p,n)\coloneqq v+\int_v^{\overline{v}}\frac{H(u\mid p,n)}{H(v\mid p,n)}du.\label{eq:bidding} \footnote{See, e.g., \citet[Section 2.3]{krishna2010auction}. In this case, the equilibrium bidding function $\beta(\cdot\mid p,n)$ solves the differential equation $d(\beta(v\mid p,n)\cdot H(v\mid p,n))/dv=v\cdot H'(v\mid p,n)$, where $H'(\cdot \mid p,n)$ denotes the derivative of $H(\cdot \mid p,n)$, with a boundary condition $\beta(\overline{v}\mid p,n)=\overline{v}$.}
\end{equation}
When the entry probability is $p$, a bidder with a signal $S=s$ has ex-ante expected profit from entry
\begin{equation*}
   \Pi(p,n,\kappa,s)\coloneqq  \int_{\underline v}^{\overline{v}}(\beta(v\mid p,n)-v)H(v\mid p,n) dF_{V\mid S}(v\mid s)-\kappa. 
\end{equation*}
The marginal entrant is characterized by the signal $ s = p $. In equilibrium, this entrant is indifferent between entering and not entering the auction, which implies that their ex-ante expected profit is zero. 

Let $p(n,\kappa)$ denote the pure-strategy symmetric equilibrium entry probability in auctions with $n$ potential bidders and entry cost $\kappa$. We have:
\begin{equation}\label{eq:equilibrium condition}
    \Pi\big(p(n,\kappa),n,\kappa,p(n,\kappa)\big)=0.
\end{equation}
The following result provides a more detailed characterization of the equilibrium entry probability.

\begin{proposition}\label{p:entry}
    Suppose that Assumption \ref{a:copula_1} holds.  If $p(n,\kappa)>0$,  then it satisfies:
    \begin{equation}\label{eq:entry}
        \int_{\underline v}^{\overline{v}} C_2(F(v),p(n,\kappa))H(v\mid p(n,\kappa), n)dv=\kappa.
    \end{equation}
\end{proposition}

The proof of the proposition establishes that, given the entry probability $p$, the expected profit from entry for a bidder with signal $S=s$ is given by 
\begin{equation}\label{eq:Pi_bid_requirement}
\Pi(p,n,\kappa,s)=\int_{\underline v}^{\overline v}C_2(F(v),s)H(v\mid p,n)dv-\kappa,
\end{equation}
which is a decreasing function of $s$ because of the ``good news'' assumption $C_{22}(\cdot,\cdot)< 0$. This confirms that the equilibrium entry strategy is to enter when $S\leq p(n,\kappa)$. 

In a model with no bid requirements, \citet{MSX} show that the marginal entrant's expected profit is non-increasing in the entry probability $p$, which results in a unique symmetric entry equilibrium. The bid requirement for at least two active bidders changes the situation in two ways. First, $p(n,\kappa)=0$ (i.e., no entry) is always an equilibrium. Second, it is evident from \eqref{eq:Pi_bid_requirement} that the expected profit from entry can now be non-monotone in the entry probability,\footnote{This is due to the $(1-p)^{n-1}$ term in the probability of winning function $H(\cdot\mid p, n)$.} which may create multiple non-trivial entry equilibria.
Intuitively, the non-monotonicity can be understood as follows: to win the contract, a bidder needs at least one other active bidder to avoid the contract's cancellation. As a result, the expected profit increases with the entry probability $p$ when $p$ is small. However, when the entry probability is sufficiently large, a bidder faces more competition at the bidding stage and is more likely to lose.

\begin{figure}
    \centering
    \begin{subfigure}{0.45\textwidth}
        \includegraphics[width=\linewidth]{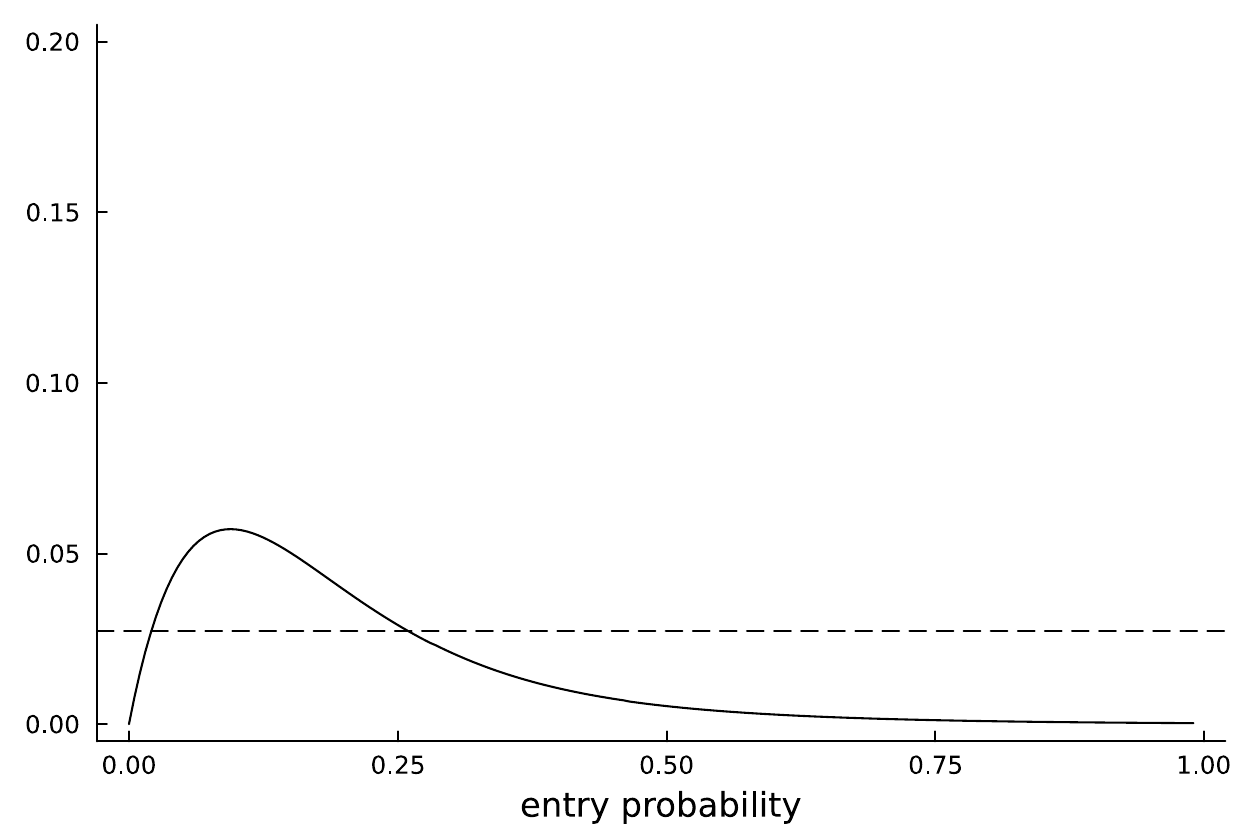}
        \caption{Bid requirement format}\label{fig:marginal_revenue_BR_14}
    \end{subfigure}
    \hspace{1ex}
    \begin{subfigure}{0.45\textwidth}
        \includegraphics[width=\linewidth]{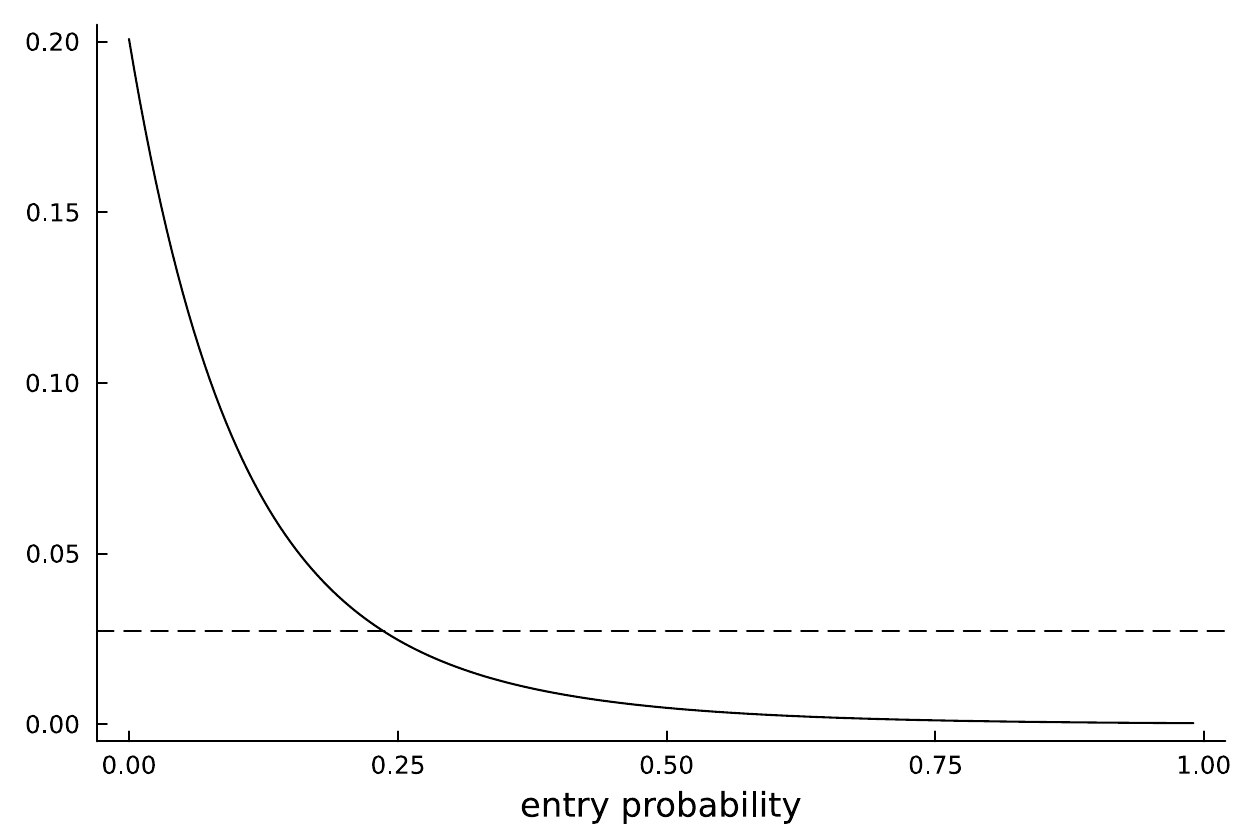}
        \caption{Reserve price format}\label{fig:marginal_revenue_RP_14}
    \end{subfigure}
    \caption{The marginal entrant's estimated expected revenue (solid line) and entry cost (dashed line), as fractions of the engineer's estimate, under the two formats for different entry probabilities, estimated using the TxDoT data for auctions with 14 potential bidders}\label{fig:multiple_equilibria}
\end{figure}

Consider the marginal entrant (i.e., an active bidder with $s=p$) and note that the type of the marginal entrant varies with $p$. By Proposition \ref{p:entry}, the derivative of the expected profit function $p\mapsto\Pi(p,n,\kappa,p)$ is given by
\begin{align*}
  \int_{\underline v}^{\overline{v}}C_{22}(F(v),p)H(v\mid p,n)dv -(n-1)\int_{\underline v}^{\overline{v}}C_2^2(F(v),p)\Lambda^{n-2}(v\mid p) dv \\ +(n-1)(1-p)^{n-2}\int_{\underline v}^{\overline{v}}C_2(F(v),p)dv. 
\end{align*}
The expression in the first line is negative by the ``good news'' assumption $C_{22}(\cdot,\cdot)< 0$. However, the positive term in the second line may dominate the expression in the first line, especially for smaller values of $p$. In such cases, the marginal entrant's expected profit from entry increases in $p$. Note that the term in the second line is present only because of the bid requirement condition.

Figure \ref{fig:marginal_revenue_BR_14} illustrates the non-monotonicity of the resulting function, showing the estimated entry cost and the expected counterfactual revenue of the marginal entrant for different entry probabilities in TxDoT procurement auctions with 14 potential bidders; see Section \ref{sec:empirical} for details. The non-monotonicity of the marginal entrant's expected profit creates multiple equilibria for entry. Besides the trivial zero-entry equilibrium, there are two other equilibria with entry probabilities $p=0.02$ and $0.26$. However, the equilibrium with $p=0.02$ is unstable, as small negative shocks to the entry probability push it away from $0.02$ toward zero; similarly, small positive shocks push the entry probability upward. Since the trivial zero-entry equilibrium generates no data and the unstable equilibrium can be ruled out, we assume that the stable equilibrium with entry probability $0.26$ generates all data in this example.

Because the marginal entrant's type is $s = p$, the type varies along the curve in Figure~\ref{fig:marginal_revenue_BR_14}. Figure~\ref{fig:eqm_R2} illustrates how the marginal entrant's type is pinned down in equilibrium. For $n = 14$, it plots the expected revenue from entry as a function of the entry probability $p$ for three signal realizations $s$. The curve corresponding to $s = 0.26$ intersects the entry cost $\kappa$ exactly at $p = 0.26$, confirming that this signal realization is the equilibrium type of the marginal entrant.

The existence of a stable non-trivial equilibrium is not guaranteed in general. As Figure \ref{fig:marginal_revenue_BR_14} illustrates, when the entry cost is sufficiently high, the expected revenue of the marginal entrant falls below the entry cost for all entry probabilities $p$, and only the trivial zero-entry equilibrium exists. Conversely, when the entry cost is sufficiently low, only the trivial and unstable equilibria exist. To see this, note that the expected revenue from entry is non-negative, and the derivative of $p\mapsto\Pi(p,n,\kappa,p)$ at $p=0$ and the limit of $\Pi(p,n,\kappa,p)$ as $p\to1$ are strictly positive. However, as Figure \ref{fig:marginal_revenue_BR_14} again illustrates, when the non-trivial stable equilibrium exists, its uniqueness is guaranteed by the quasi-concavity of $p\mapsto\Pi(p,n,\kappa,p)$. Under mild regularity conditions on the copula density function, we show below that for sufficiently large $n$, the function $p\mapsto\Pi(p,n,\kappa,p)$ is strictly decreasing except in a small neighborhood of zero.

\begin{figure}
  \centering
  \includegraphics[width=0.5\textwidth]{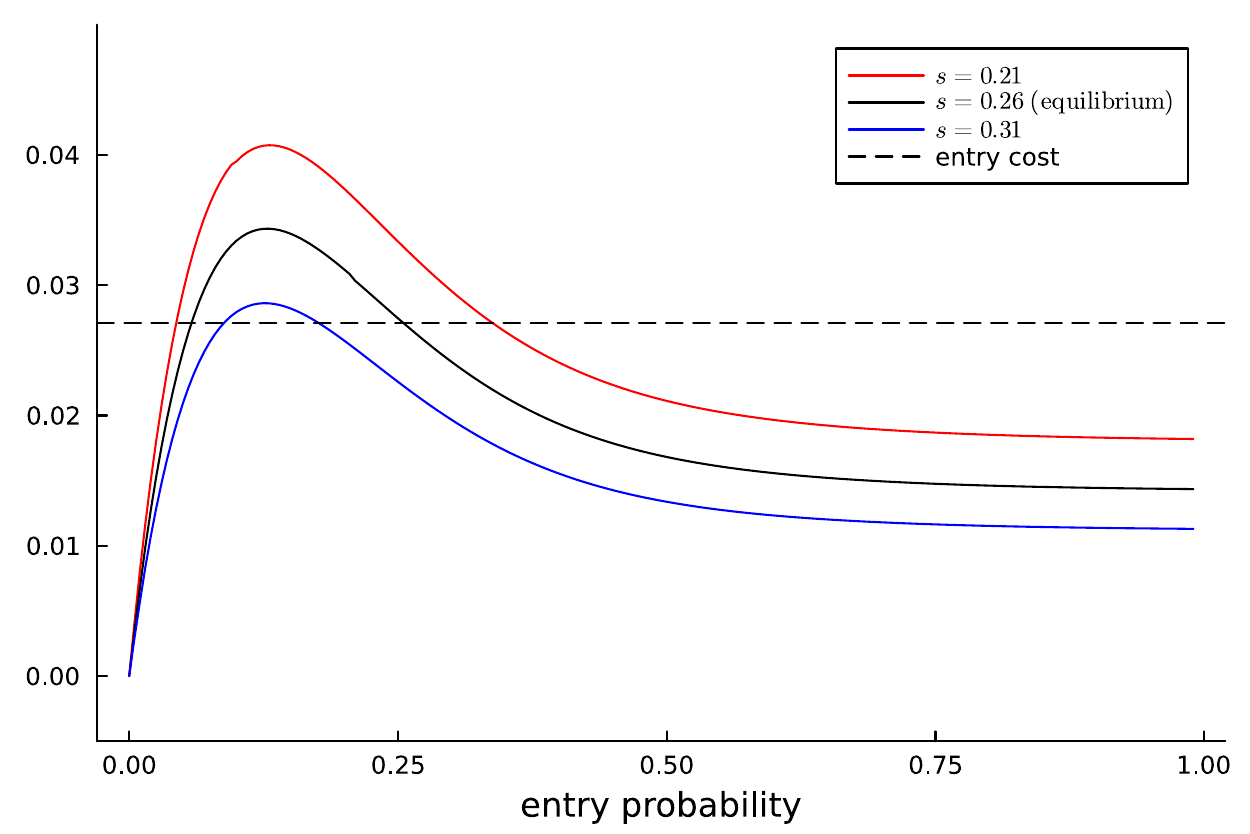}
  \caption{Expected revenue curves for three different signal realizations $s=0.21$ (red), $0.26$ (black), and $0.31$ (blue). The middle curve corresponding to $s = 0.26$ intersects the entry cost $\kappa$ at $p = 0.26$, establishing the equilibrium type of the marginal entrant.}
  \label{fig:eqm_R2}
\end{figure}

\begin{proposition}\label{p:quasi-concavity}
Suppose that Assumption \ref{a:copula_1} holds and the copula density function is bounded away from zero on $[0,1]^2$. Then, for any fixed $\varepsilon\in\left(0,1\right)$ and all sufficiently large $n$, the function $p\mapsto \Pi(p,n,\kappa,p)$ is strictly decreasing on $\left[\varepsilon,1\right]$.
\end{proposition}

Proposition \ref{p:quasi-concavity} implies that for any arbitrarily small $\varepsilon\in\left(0,1\right)$ and all sufficiently large $n$, there exists at most one stable equilibrium entry probability in $\left[\varepsilon,1\right]$. In practice, this property may hold for all $n\ge 2$. Given estimates of $F(\cdot)$ and the copula, one can verify quasi-concavity by estimating the derivative of $p\mapsto\Pi(p,n,\kappa,p)$ and checking that it changes sign at most once. For our TxDoT-based estimates, quasi-concavity holds for all $n\ge 2$.

In what follows, we assume that the function $p\mapsto\Pi(p,n,\kappa,p)$ is quasi-concave. Furthermore, let $p(n,\kappa)$ denote the non-trivial stable equilibrium when it exists. When the stable non-trivial equilibrium does not exist, we set $p(n,\kappa)=0$.

Given the equilibrium entry and bidding strategies, we can now describe the expected cost of procurement. Since the procurement cost is undetermined when an auction fails (i.e., fewer than two active bidders), we condition on receiving at least two bids. Let $N^*$ denote the number of active bidders.
\begin{proposition}\label{p:cost_conditional}
Suppose that Assumption \ref{a:copula_1} holds. Conditional on the number of active bidders $N^*\geq 2$ and entry probability $p$, the expected winning bid is given by
\begin{equation}
\label{eq:cost_conditional}
\begin{split}
    K(p, n\mid N^*\geq 2) \coloneqq & \frac{1}{\Pr[N^*\geq 2\mid p,n]}\Bigg(n \int_{\underline{v}}^{\overline{v}} \Lambda^{n-1}(v\mid p)\left(1-\frac{n-1}{n}\Lambda(v\mid p)\right)dv \\
    & +\underline v - \overline{v} \cdot \Pr[N^*<2 \mid p,n]\Bigg),
\end{split}
\end{equation}
where 
\begin{equation}
\label{eq:N^* geq 2}
    \Pr[N^*\geq 2\mid p,n]\coloneqq 1-(1-p)^n-np(1-p)^{n-1}
\end{equation}
is the probability of having at least two active bidders given the entry probability $p$.
\end{proposition}

Two equilibrium objects are of particular interest in our empirical analysis. First, the probability of auction failure (i.e., the event that the auction receives at most one bid) is given by $\Pr[N^*<2\mid p(n,\kappa),n]$. Second, the expected equilibrium winning bid conditional on there being at least two bidders, $K(p(n,\kappa), n\mid N^*\geq 2)$, corresponds to the price in \cite{li2009entry}. 

Some remarks regarding these equilibrium objects are in order. Auction failures are often not directly observed in data that contain only award files, as is the case in our empirical application. If records of call-for-bids are publicly available, merging them with the award files yields procurement data that contain records of failed auctions.\footnote{Government agencies are often legally required to advertise the projects and call-for-bids.} Even when the data contain such records, the reasons for failure may not be reported. For example, a review of some of the TxDoT archived letting outcomes in 2024 shows that all unawarded projects are listed without any indication of why they were not awarded. Multiple factors may cause a project to go unawarded: bids that are too high, non-compliant, or non-responsive; changes in priorities, funding problems, etc. Thus, in the absence of detailed institutional information, even an unawarded outcome recorded in the data may reflect multiple mechanisms beyond low entry or inadequate competition, which complicates the analysis. In such cases, using only data on successful auctions, the expression for $\Pr[N^* < 2\mid p,n]$ from Proposition \ref{p:cost_conditional} can be used to estimate the probability of auction failure and conduct counterfactual analysis (see Section \ref{subsec:identification} for a detailed discussion).
 

If one is concerned with the unconditional expected procurement cost, it would require making an assumption about the cost of failure for the auctioneer. For example, suppose that when an auction fails and the project cannot be delayed, the auctioneer hires a contractor at the maximum cost $\overline v$. The corresponding unconditional expected procurement cost is then given by
\begin{equation*}
    K(p, n)=n\int_{\underline{v}}^{\overline{v}} \Lambda^{n-1}(v\mid p)\left(1-\frac{n-1}{n}\Lambda(v\mid p)\right)dv +\underline v.
 \end{equation*}
The expression is convenient for counterfactual analysis because it does not involve additional parameters. Nevertheless, it requires the strong assumption that the auctioneer's cost of failure is $\overline v$, which may be difficult to justify in practice. In particular, invoking the outside option typically requires a strong justification that increased competition is infeasible even after additional marketing efforts, in addition to demonstrating urgency in completing the project. Therefore, in the subsequent analysis, we focus on the expected winning bid conditional on receiving at least two bids.


 
\subsection{Reserve price format\label{sec:reserve_price}}
Under this format, the contract is awarded if at least one bid is below the reserve price $r$. An auction fails when no bidder enters or when all active bidders draw values above the reserve price.


Since the contract is awarded even with a single bid below the reserve price, an active bidder with a private cost $v\leq r$ wins the auction with probability $\Lambda^{n-1}(v\mid p)=(1-C(F(v),p))^{n-1}$. In this format, the probability of winning decreases with $p$ since there is no bid requirement. For $v\leq r$, the bidding function is given by
\begin{align*}
    \beta(v\mid p,n,r)\coloneqq v+\int_{v}^{r}\left(\frac{\Lambda(u\mid p)}{\Lambda(v\mid p)}\right)^{n-1}du.
\end{align*} 
Using the same arguments as in the proof of Proposition \ref{p:entry}, when the entry cutoff is $p$, the expected profit from entry for a bidder with $S=s\leq p$ is given by 
\begin{align}\label{eq:Pi_reserve_price}
   \Pi(p,n,\kappa,r,s)\coloneqq \int_{\underline v}^{r}C_2 (F(v),s)\Lambda^{n-1}(v\mid p)dv-\kappa,
\end{align}
and the pure-strategy symmetric equilibrium probability of entry $p(n,\kappa,r)$ solves
\begin{align*}
    \Pi(p(n,\kappa,r),n,\kappa,r,p(n,\kappa,r))&=0.
\end{align*}
In this case, the entry equilibrium is unique, as the expected revenue of the marginal entrant (a bidder with $S=p$) is decreasing in $p$, as in \citet{MSX}; see Figure \ref{fig:marginal_revenue_RP_14}. 

Comparing Figures \ref{fig:marginal_revenue_BR_14} and \ref{fig:marginal_revenue_RP_14} provides the central insight into the differences in entry between the two formats. Under the reserve price format, the marginal entrant's expected revenue is monotonically decreasing in the entry probability and takes relatively large values for small entry probabilities. This is expected, since for small entry probabilities the marginal entrant would face fewer competitors or may even be the sole active bidder. Under the bid requirement format, by contrast, the contract is likely to be canceled for small entry probabilities, which substantially reduces the marginal entrant's expected revenue. As a result, the reserve price format can support a broader range of entry costs with strictly positive entry probabilities.

Note that in this low-price auction setting, the bidder's expected revenue from entry is an increasing function of the reserve price $r$. Since $\partial \Lambda(v\mid p)/\partial p<0$ and $C_{22}(F(v),p)< 0$, it follows that a lower reserve price corresponds to a smaller equilibrium entry probability. 

With a binding reserve price $\underline v < r<\overline{v}$, a bidder is active if their signal is below the entry cutoff $p$ and their value is below the reserve price $r$. This occurs with probability $0<C(F(r),p)<p$.  Given the entry cutoff $p$, the probability of auction failure (i.e., not receiving any bids) is
\begin{equation*}
    \Pr[N^*=0\mid p,n,r]\coloneqq (1-C(F(r),p))^{n},
\end{equation*}
where $N^*$ again denotes the number of active bidders. Besides the distribution of values and signals and the number of potential bidders, the probability of failure also depends on the reserve price. Since a lower reserve price implies a smaller entry probability in equilibrium, a lower reserve price also increases the auction failure probability.

 By the same arguments as in the proof of Proposition \ref{p:cost_conditional}, we can show that the expected winning bid conditional on $N^*\geq 1$ and entry probability $p$ is given by
\begin{align}
    K(p,n, r\mid N^*\geq 1) &\coloneqq
  \frac{1}{\Pr[N^*\geq 1 \mid p,n,r]}\Bigg(n \int_{\underline{v}}^{r} \Lambda^{n-1}(v\mid p)\left(1-\frac{n-1}{n}\Lambda(v\mid p)\right)dv \notag\\
    &\qquad +\underline v - r\cdot  \Pr[N^*=0 \mid p,n,r]\Bigg).\label{eq:win_Bid_reserve}
\end{align}
The equilibrium probability of auction failure and expected winning bid are given by $\Pr[N^*=0\mid p(n,\kappa,r),n,r]$ and $K(p(n,\kappa,r),n,r\mid N^*\geq 1)$, respectively. As in the bid requirement format, a larger probability of entry corresponds to a lower probability of auction failure.

The analytical comparison of the two formats does not predict which has a higher probability of success. Fixing the entry cutoff $p$ for both formats, the difference in the probabilities of auction failure between the bid requirement and reserve price formats is given by
\begin{equation*}
    \Big((1-p)^n-(1-C(F(r),p))^n\Big)+np(1-p)^{n-1}.
\end{equation*}
The first term (the difference between the probabilities of zero active bidders) is negative. However, the second term (the probability of only one active bidder under the bid requirement format) is positive, and the comparison is ambiguous. Moreover, the equilibrium entry probabilities differ between the formats. For example, in the empirical section below, we find that the probability of entry can be smaller or larger under the reserve price format in the TxDoT setting, depending on the number of potential bidders. Nevertheless, even when the probability of entry is larger under the reserve price format, the probability of bidding is lower than under the bid requirement format. The analytical comparison of the expected winning bids is similarly ambiguous.

\section{Role of signal informativeness} \label{sec:Signals_cost}

The level of signal informativeness plausibly varies in practice. In this section, we provide analytical results on how signal informativeness interacts with auction formats to affect procurement outcomes. To study this interaction, we model the joint distribution of private costs and signal ranks using a parametric family of copula functions indexed by a scalar parameter that captures the level of signal informativeness.\footnote{The parametric copula assumption will continue to play an important role for the identification and estimation of the model in the sections below because signals are unobserved.} Specifically, we assume that the joint distribution of private costs and signal ranks is given by
\begin{equation}
    C(F(v),p)=C(F(v),p;\theta_0),
\label{eq:parametric_copula}
\end{equation}
where the function $C(\cdot,\cdot;\theta)$ is known up to the value of a scalar parameter $\theta\in\Theta\subset\mathbb R$ and $\Theta$ denotes the set of parameters permitted by the chosen copula function. Note that the marginal distribution of private costs $V$ remains nonparametric. This semiparametric approach is convenient, as the single parameter $\theta$ captures the dependence between the private costs $V$ and the uniformly distributed signal ranks $S$. Therefore, $\theta$  can be viewed as a measure of signal informativeness.
We make the following assumption.
\begin{assumption}\label{a:copula}
    \begin{enumerate}[(i)]
    \item[] 
    \item 
    $C(x,y;\cdot)$  is continuously differentiable with $\partial C(x,y;\theta)/\partial \theta \geq 0$. \label{a:positive_ordering}
    \item Assumptions \ref{a:copula_1}\ref{a:C_diffble}--\ref{a:goodnews}  are satisfied by $C(x,y;\theta)$ for all $\theta\in\Theta$.
    \end{enumerate}
\end{assumption}

Under Assumption \ref{a:copula}\ref{a:positive_ordering}
, the family of copulas $\{C(x,y;\theta):\theta \in \Theta\}$ is positively ordered: for all $x,y\in[0,1]$ and all $\theta_1\leq \theta_2$, $C(x,y;\theta_1)\leq C(x,y;\theta_2)$. Many copula families satisfy the positive ordering assumption, including the Gaussian copula and important Archimedean copulas such as Ali-Mikhail-Haq, Clayton, Frank, Gumbel, and Joe. For such copulas, a higher value of $\theta$ corresponds to a stronger association between private costs and signals, as measured by statistics such as Kendall's $\tau$ or Spearman's $\rho$ \citep[][Chapter 5]{nelsen2007introduction}. In our auction context, the positive ordering property thus ensures that higher values of $\theta$ correspond to more informative signals.

The positive ordering property also has implications for the distribution of private costs conditional on entry, as defined in equation \eqref{eq:F^*}. For a given entry probability $p$, more informative signals make this conditional distribution less stochastically dominant, so that entering bidders tend to have lower costs.

As discussed in Section \ref{sec:model}, increased entry reduces the probability of auction failure under both formats. Therefore, to study the effect of signal informativeness $\theta$ on auction failure, we can focus on its effect on the equilibrium probability of entry.\footnote{Note that under the reserve price format, the probability of auction failure depends on the entry probability as well as the direct probability that an entrant's cost falls below the reserve price.} To this end, consider how $\theta$ affects the marginal entrant's expected revenue from entry (see Figure \ref{fig:multiple_equilibria}). If, for instance, the marginal entrant's expected revenue from entry decreases in $\theta$ at all relevant entry probabilities $p$, then the stable equilibrium entry probability under the bid requirement format would also decrease in $\theta$.

Define $\Lambda(v\mid p;\theta)\coloneqq 1-C(F(v),p;\theta)$ and $H(v\mid p,n;\theta)\coloneqq \Lambda^{n-1}(v\mid p;\theta)-(1-p)^{n-1}$.
The equilibrium entry probabilities under the bid requirement format can now be written as $p(n,\kappa;\theta_0)$. Other functions are similarly augmented with the copula parameter to reflect their dependence on signal informativeness.
By Proposition \ref{p:entry}, under the bid requirement format, the derivative of the marginal entrant's expected revenue from entry with respect to $\theta$ is 
\begin{align*}
\lefteqn{\int_{\underline v}^{\overline v}\frac{\partial C_2(F(v),p;\theta)}{\partial \theta}H(v\mid p,n;\theta) dv} \\
&\quad -(n-1)\int_{\underline v}^{\overline v}C_2(F(v),p;\theta)\Lambda^{n-2}(v\mid p;\theta)\frac{\partial C(F(v),p;\theta)}{\partial \theta}dv.
\end{align*}
By the positive ordering condition of Assumption \ref{a:copula}\ref{a:positive_ordering}, the expression in the second line is negative. However, the first term can be positive or negative. To see why, note that $C_2(F(v),p;\theta)$ is the conditional CDF of private costs given $S=p$. The positive ordering condition implies that,
as $\theta$ increases, this conditional CDF concentrates around the
value of $v$ satisfying $F(v)=p$, with the conditional CDF curves
for different $\theta$'s crossing at that point. In other words, more informative signals do not imply a stochastic dominance relation on the conditional CDFs of $V$ given $S=p$. Therefore, the theory does not predict whether more informative signals lead to a lower or higher probability of entry under the bid requirement format. The situation is similar under the reserve price format, with no definite prediction for the effect of $\theta$ on equilibrium entry.

Nevertheless, we show below that under the bid requirement format, entry stops completely when signal informativeness $\theta$ is sufficiently high. By contrast, the reserve price format can support positive entry probabilities even when signals are perfectly informative. Under the bid requirement format, let $\Pi(p,n,\kappa,s;\theta)$ denote the expected profit from entry of a potential bidder with signal $s$ in auctions with $n$ potential bidders and entry cost $\kappa$, when the entry probability is $p$ and signal informativeness is $\theta$. Similarly, under the reserve price format, we use $\Pi(p,n,\kappa,r,s;\theta)$ to denote the same object in auctions with a reserve price $r$. Recall that the marginal entrant is characterized by $s=p$. Finally, suppose that as $\theta$ increases, private costs and signals become perfectly positively dependent; that is, $\lim_{\theta\uparrow \infty}C(x,y;\theta)=\min\{x,y\}$, where the limiting function is the comonotonicity copula.
\begin{proposition}\label{prop:limit_theta_infty}
    Suppose that Assumptions \ref{a:copula_1}\ref{a:marginal_cdf} and \ref{a:copula} hold. Furthermore, suppose that  $\lim_{\theta\uparrow \infty}C(x,y;\theta)=\min\{x,y\}$. The expected profit of the marginal entrant satisfies:
    \begin{enumerate}[(a)]
        \item Under the bid requirement format, \[\lim_{\theta \uparrow \infty} \Pi(p,n,\kappa,p;\theta) = -\kappa.\] \label{prop:limit_theta_bid_req}
        \item Under the reserve price format, \[\lim_{\theta \uparrow \infty} \Pi(p,n,\kappa,r,p;\theta) = \mathbbm{1}(F(r)\geq p)(r-F^{-1}(p))(1-p)^{n-1}-\kappa,\] \label{prop:limit_theta_res_price} where $\mathbbm{1}(\cdot)$ denotes the indicator function.
    \end{enumerate}
\end{proposition} 

Part \ref{prop:limit_theta_bid_req} of the proposition shows that when signals are sufficiently informative, the marginal entrant's expected revenue under the bid requirement format is zero for all values of $p$. Consequently, the marginal entrant's profit is negative and entry stops. Part \ref{prop:limit_theta_res_price}, by contrast, shows that the marginal entrant's expected revenue can be positive under the reserve price format, provided that the entry probability $p$ is sufficiently small relative to the reserve price, as captured by the $r-F^{-1}(p)$ term.

The results of Proposition \ref{prop:limit_theta_infty} have an intuitive interpretation. When signals are highly informative, the marginal entrant is likely to have the worst private cost among all entering bidders and therefore can only win when there are no other active bidders. Under the bid requirement format, however, the auction would be canceled in such a scenario. Since the marginal entrant cannot win the contract, entry unravels. By contrast, under the reserve price format, the marginal entrant can win the contract with no other active bidders, so positive entry can be sustained even when signals are perfectly informative.

While the bid requirement format leads to certain auction failure when signals are sufficiently informative, it may nonetheless have an advantage over the reserve price format when signal informativeness is low. In such cases, marginal entrants may have competitive private costs and can win the contract even when other active bidders are present. For example, entry is completely random when signals and private costs are independent, as in the \citet{levin1994equilibrium} model. By comparison, entry can be more restricted under the reserve price format if the reserve price is set too low, leading to a higher probability of auction failure.

Turning to the cost of procurement, let $K(p, n\mid N^*\geq 2;\theta)$ denote the expected winning bid under the bid requirement format as defined in Proposition \ref{p:cost_conditional}, with explicit dependence on the signal informativeness parameter $\theta$. Similarly, $p(n,\kappa;\theta)$ denotes the equilibrium entry probability in the stable non-trivial equilibrium, provided it exists.
In equilibrium, changing $\theta$ has two effects on the expected winning bid, which we refer to as the ``information'' and ``cutoff'' effects:
\begin{align*}
\lefteqn{\frac{\partial K(p(n,\kappa;\theta), n\mid N^*\geq 2;\theta)}{\partial \theta}  =}\\
 &\quad \underbrace{\frac{\partial K(p, n\mid N^*\geq 2;\theta)}{\partial \theta}\Bigg|_{p=p(n,\kappa;\theta)}}_{\text{information effect}} 
+\underbrace{\frac{\partial K(p(n,\kappa;\theta), n\mid N^*\geq 2;\theta)}{\partial p } \frac{\partial p(n,\kappa;\theta)}{\partial \theta}}_{\text{cutoff effect}}.
\end{align*}
The information effect is the direct impact of more informative signals. It operates through the function $\Lambda(v\mid p;\theta)$ and the distribution of private costs conditional on entry. By the positive ordering in Assumption \ref{a:copula}\ref{a:positive_ordering}, the latter is less stochastically dominant under more informative signals. In other words, entering bidders tend to have smaller private costs, which reduces the expected winning bid:
\begin{align*}
   \lefteqn{\frac{\partial K(p, n\mid N^*\geq 2; \theta)}{\partial \theta} =}\\  &\quad-\frac{n(n-1)}{\Pr[N^*\geq 2\mid p,n]}\int_{\underline v}^{\overline{v}} \frac{\partial C(F(v),p;\theta)}{\partial \theta}\Lambda^{n-2}(v\mid p;\theta)C(F(v),p;\theta)dv\leq 0.
\end{align*}

On the other hand, the cutoff effect is ambiguous because signal informativeness can have a positive or negative effect on equilibrium entry. Comparative statics cannot predict the direction of this impact.

The comparative statics for the expected winning bid under the reserve price format are similar to those for the bid requirement format. The information effect reduces the expected cost of procurement, but the cutoff effect is ambiguous because the equilibrium entry probability can increase or decrease with signal informativeness.

\section{Semiparametric identification and estimation}\label{sec:ID}
\subsection{Identification}\label{subsec:identification}
We estimate the model using the data from \citet{li2009entry}, assuming the data were generated under the bid requirement format. Therefore, we focus on identifying the model's primitives under this format. 

Since signals are unobserved and there are no exogenous entry-cost shifters, the joint distribution of signals and private costs cannot be identified nonparametrically. Therefore, to restore full identification of the model primitives, we continue to use the semiparametric specification in equation \eqref{eq:parametric_copula}, treating the marginal distribution of private costs nonparametrically while modeling the joint distribution of private costs and signals using a parametric copula function.\footnote{The semiparametric approach used in this paper was first proposed in \cite{GL}, see the discussion on page 332 therein, including footnote 18. Note that in our application, there are no exogenous cost shifters that could be used for a full non-parametric identification, as proposed in \cite{GL}.}  Note that from the signal component of the model, only the probability of entry is required to identify the private costs, the entry cost, and the parameter $\theta$. The knowledge or specification of the marginal distribution of the underlying signals is not required.

Suppose that the econometrician observes data from $L$ independent auctions. We use $N_l$ and $N^*_l$ to denote the observed (random) numbers of potential and active bidders, respectively, in auctions $l=1,\ldots,L$. For each auction $l$, we observe equilibrium bids $B_{1l},\ldots,B_{N_l^* l}$. As in \citet{MSX}, the main source of identification of the distribution of private costs and the copula parameter is the variation in the number of potential bidders. Therefore, we assume that the number of potential bidders is exogenous. Furthermore, we assume that auctions with the same number of potential bidders have the same entry cost.\footnote{Unlike \citet{MSX}, we do not assume that all auctions have the same entry cost regardless of the number of potential bidders.} We now introduce assumptions on the data-generating process of the equilibrium bids and entry decisions.
\begin{assumption}
  \label{a:ID1}
\begin{enumerate}[(i)]
\item[]
\item An auction,  including the bids of the active bidders and the number of potential bidders, is observed if and only if  the number of active bidders is no less than 2.\footnote{The assumption is equivalent to assuming that the bid requirement condition holds independently of  other potential reasons for canceling auctions.}\label{a:observed_auctions}
\item Auctions with the same number of potential bidders $N_l=n$ have the same entry cost $\kappa_n$ for all $n\in\mathcal N$, where $\mathcal N$ denotes the support of the distribution $N_l$. \label{a:same_entry_cost}
\item The signal ranks $S_{il}$ of potential bidders $i=1,\ldots,N_l$ in auction $l$ are independently distributed across $i$ and $l$  and are independent of the number of potential bidders $N_l$.\label{a:independent}
\item The function $p\mapsto \Pi(p,n,\kappa_n,p;\theta_0)$ is quasi-concave for all $n\in\mathcal N$. \label{a:unique_equilibrium}
\item The observed number of active bidders $N_{l}^{*}$  is drawn from the conditional distribution of $\bar{N}_{l}\coloneqq\sum_{i=1}^{N_{l}}\mathbbm{1}\left(S_{il}\leq p(N_l,\kappa_{N_l};\theta_0)\right)$  given $\bar{N}_{l}\geq2$, where $p(n,\kappa_n;\theta_0)$ is the equilibrium entry probability as defined in Proposition \ref{p:entry} with $C(\cdot,\cdot)=C(\cdot,\cdot;\theta_0)$.\label{a:N_star dgp}
\item The private costs $V_{il}$ of active bidders $i=1,\ldots,N^*_l$ in auction $l$ are independent draws from the conditional distribution $F^*(\cdot\mid p_{N_l})=C(F(\cdot),p_{N_l};\theta_0)/p_{N_l}$. \label{a:F_star}
\item The CDF $F(\cdot)$ satisfies Assumption \ref{a:copula_1}\ref{a:marginal_cdf}.
    \end{enumerate}
\end{assumption}

Note that since auctions with the same number of potential bidders $n$ are assumed to have the same entry cost $\kappa_n$ in Assumption \ref{a:ID1}\ref{a:same_entry_cost}, we can now write the equilibrium entry probability as
\[p_n \coloneqq p(n,\kappa_n;\theta_0).\]
Moreover, since $\kappa_n$ is a deterministic function of $n$ and the distribution of private costs and signals is independent of the number of potential bidders by Assumption \ref{a:ID1}\ref{a:independent}, we effectively assume that the private costs and signals are independent of the entry cost.

Consider auctions with $N_l=n$ potential bidders. Since we observe the numbers of potential and active bidders, $\E[N^*_{l}/N_{l}\mid  N_l=n]$ is directly identified from the data. Recall from Assumption \ref{a:ID1}\ref{a:N_star dgp} that $\bar{N}_l$ is the latent count of potential bidders whose signals satisfy the entry condition, and that $N^*_l$ has the conditional distribution of
$\bar{N}_l$ given the bid-requirement event $\bar{N}_l\geq 2$. 
We have:
\begin{equation}\label{eq:p_N_ID}
  \E\left[\frac{N^*_{l}}{N_{l}}\mid N_l=n\right] = \Pr[S_{1l}\leq p_n\mid \bar{N}_{l}\geq 2, N_l=n] = \frac{p_n\left(1-\left(1-p_n\right)^{n-1}\right)}{1-\left(1-p_n\right)^{n}-np_n\left(1-p_n\right)^{n-1}},\end{equation}
  where the first equality is by symmetry across potential bidders, and the second equality is derived in the Appendix. 
When the number of potential bidders is two, auctions are observed only if both bidders enter, so any entry probability $0<p_2<1$ is observationally equivalent to entry with probability one and $p_2$ is unidentified. However, for all $n\geq 3$, $p_n$ is identified using equation \eqref{eq:p_N_ID}, as we show in Proposition S1 in the Supplement.

If call-for-bids records, reasons for failure, and the award files with bids are all observed, then $\Pr[N^* < 2\mid p_n,n]$ is an observed feature of the data. In this case, $p_n$ is also identifiable from equation \eqref{eq:N^* geq 2} and therefore $p_n$ is over-identified. This equation can be incorporated into the estimation procedure to improve efficiency and to test the model. Moreover, $p_2$ is identifiable from $\Pr[N^* < 2\mid p_2,2] = 1-p_2^2$.


The CDF of private costs conditional on entry $F^*(\cdot\mid p_n)$ is nonparametrically identified from the data using a modification of the inverse-bidding-function approach of \cite{GPV}.\footnote{In this semiparametric model, the CDF of private costs conditional on entry depends on $\theta$ through the parametric copula function. However, we omit $\theta_0$ from the notation $F^*(v\mid p_n)$ to emphasize that it is nonparametrically identified. } In the context of auctions with entry, this approach was used by \citet{MSX} and \cite{xu2013nonparametric}, but in our case it also requires an adjustment for the bid requirement condition $N_l^*\geq 2$.
Let $\xi(\cdot\mid p, n)$ denote the inverse bidding strategy in auctions with entry probability $p$ and $N_l=n$ potential bidders:\footnote{We suppress $\theta_0$ from the notation for the inverse bidding function as it is identified nonparametrically. Following the same logic, we suppress $\theta_0$ from the notation for the CDF of bids, etc.} \[\xi(\cdot\mid p,n)\coloneqq \beta^{-1}(\cdot\mid p,n).\]
We use $G(\cdot\mid n)$ and $g(\cdot \mid n)$ to denote the CDF and PDF of the submitted bids, respectively, in auctions with $N_l=n$ potential bidders. Both functions can be estimated consistently from bid data. The CDF of private costs conditional on entry satisfies $F^*(v\mid p_n) 
= G(\xi^{-1}(v \mid p_n, n)\mid n)$ and, therefore, is identified if the inverse bidding function $\xi(\cdot\mid p_n,n)$ is identified. Consider an auction in which each of the $n-1$ competitors enters with probability $p$ and bids according to $G(\cdot \mid n)$. Let
\[W\left(b \mid p,n \right)\coloneqq \left\{ 1-p\cdot G\left(b\mid n\right)\right\} ^{n-1}-\left(1-p\right)^{n-1}
\]
denote the probability of winning with bid $b$, and let $W'\left(b \mid p,n \right)\coloneqq \partial W\left(b \mid p,n \right)/\partial b$ denote its derivative with respect to $b$. Since $\beta (v\mid p_n,n)$ is the best response, it must satisfy the first-order condition
\[W(\beta (v\mid p_n,n) \mid p_n, n)+(\beta (v\mid p_n,n)-v)W'(\beta (v\mid p_n,n) \mid p_n,n)=0.\]
The following result shows that the inverse bidding function is nonparametrically identified from bid data.

\begin{proposition}\label{p:inverse_bidding_function}
  Suppose that Assumption \ref{a:ID1} holds.  The inverse bidding function satisfies 
     \begin{align}\label{eq:inverse_bidding_function}
\xi\left(b\mid p_n, n\right) &= b+\frac{W\left(b\mid p_n,n\right)}{W'\left(b\mid p_n,n\right)}  =  b-\frac{\eta_{n}\left(p_n,G\left(b\mid n\right)\right)}{\left(n-1\right)g\left(b\mid n\right)},\text{ where }\\
\eta_{n}\left(p,y\right)&\coloneqq \frac{1}{p}\left({1-p\cdot y-\frac{\left(1-p\right)^{n-1}}{\left(1-p\cdot y\right)^{n-2}}}\right).
    \end{align}
\end{proposition}


The copula parameter is identified through the restriction the copula imposes on the CDF of private costs conditional on entry. This restriction arises from the independence of the number of potential bidders from the private costs and signals. Let $\mathcal N$ denote the support of the distribution of the number of potential bidders $N_{l}$. We have
\begin{align}\label{eq:restrictions_for_theta_and_F}
    F^*(v\mid p_n) = {C(F(v),p_n;\theta_0)}/{p_n}, \text{ for all }v\in[\underline{v},\overline v]\text{ and }n\in\mathcal{N}.
\end{align}
Note that exogenous variation in the number of potential bidders $N_l$ alone is insufficient for identification. Since we allow the entry cost $\kappa_n$ to vary with $n$, the entry probability $p_n$ can be the same for different values of $n$, in which case equation \eqref{eq:restrictions_for_theta_and_F} does not identify the copula parameter or the distribution of private costs. A necessary condition for identification is that the equilibrium entry probabilities differ for at least two different numbers of potential bidders. This restricts all the primitives: the copula function, the CDF of private costs, the numbers of potential bidders, and the entry costs $\kappa_n$. Since equilibrium entry probabilities are identified, we can verify this condition empirically. A result in Section S1.2 of the Supplement shows that for any $\varepsilon \in (0,1)$, if the numbers of potential bidders are sufficiently large and the entry costs are constant, any pair of equilibrium entry probabilities in $[\varepsilon,1]$ must differ.

Let ${Q}(\cdot,v;\theta)$ denote the inverse function of 
 ${C}(\cdot,v;\theta)/v$. The restriction in \eqref{eq:restrictions_for_theta_and_F} can be rewritten as
\begin{equation} \label{eq:F=Q}
    F(v)={Q}(F^*(v\mid p_n),p_n;\theta_0), \text{ for all }v\in[\underline{v},\overline v]\text{ and }n\in\mathcal{N}.
\end{equation}
Since $F(\cdot)$ does not depend on $n$, the expression in \eqref{eq:F=Q} implies that for all $v\in[\underline{v},\bar v]$ and $n_1,n_2\in\mathcal{N}$,
\begin{equation}\label{eq:theta_identification_equation}
    {Q}(F^*(v\mid p_{n_1}),p_{n_1};\theta_0)={Q}(F^*(v\mid p_{n_2}),p_{n_2};\theta_0).  
\end{equation}
Proposition S2 in the Supplement provides sufficient conditions for the global and local identification of the copula parameter $\theta_0$ using the nonlinear restrictions in \eqref{eq:theta_identification_equation}. The conditions require that the equilibrium entry probabilities differ for at least two different numbers of potential bidders and that the copula function satisfies a mild regularity condition. 

With the copula parameter $\theta_0$ identified, we can recover the CDF $F(\cdot)$ of private costs from \eqref{eq:F=Q} and further recover $\Lambda(v\mid p_n;\theta_0)=1-C(F(v),p_n;\theta_0)$. The entry cost $\kappa_n$ can now be identified for all $n\in\mathcal N$ using the result of Proposition \ref{p:entry}:
\begin{equation}\label{eq:entry_cost_ID}
    \kappa_n=\int_{\underline v}^{\overline{v}} C_2(F(v),p_n;\theta_0)\Big(\Lambda^{n-1}(v\mid p_n;\theta_0)-(1-p_n)^{n-1}\Big)dv. 
\end{equation}
Using the corresponding expressions in Sections \ref{sec:bid_requirement} and \ref{sec:reserve_price}, we can similarly compute the expected winning bids and probabilities of auction failure under both formats; these quantities are used for counterfactual experiments.

\subsection{Estimation}\label{sec:estimation}

We take a semiparametric approach to estimate the model's primitives. Nonparametrically identified objects are estimated nonparametrically, while the objects identified through the parametric copula assumption are estimated using GMM based on the model's restriction in \eqref{eq:restrictions_for_theta_and_F}.

Let $\widehat G (\cdot \mid n)$ denote the empirical CDF of submitted bids in auctions with $n$ potential bidders. Following \citet{Ma2021}, we use a boundary-adaptive
local linear kernel estimator $\widehat{g}\left(\cdot \mid n\right)$ 
of the PDF of submitted bids $g\left(\cdot\mid n\right)$ in auctions with $n$ potential bidders.\footnote{
The boundary-adaptive estimator coincides with the usual kernel density estimator away from the boundaries while correcting the bias of the latter near the boundaries. Thus, it avoids trimming near-boundary observations as in \cite{GPV}. See also \citet{Hickman:2014kq}, \citet{Ma2019}, and \citet{Zincenko2024} for the discussions of the issue and various solutions in the auctions context. }
Let $\widehat p_n$ denote the estimated entry probability in auctions with $n$ potential bidders computed using the empirical analog of \eqref{eq:p_N_ID}. Using the plug-in approach and \eqref{eq:inverse_bidding_function}, we can construct an estimator of the inverse bidding function $\xi(\cdot\mid p_n,n)$:
\begin{equation*}
\widehat{\xi}\left(b\mid n\right)\coloneqq b-\frac{\eta_{n}\left(\widehat{p}_{n},\widehat{G}\left(b\mid n\right)\right)}{\left(n-1\right)\widehat{g}\left(b\mid n\right)}.
\end{equation*}
We use this estimator to compute the estimated (pseudo) private costs $\widehat{V}_{il}\coloneqq\widehat{\xi}(B_{il}\mid N_l)$ and then use the empirical CDF of the pseudo costs to estimate the CDF of private costs conditional on entry in auctions with $n$ potential bidders:\footnote{Recently, \citet{Zincenko2024} proposed a similar empirical CDF estimator for the private value CDF and studies its asymptotic linearization. However, \citet{Zincenko2024}'s estimated inverse bidding function uses the kernel-smoothed bid CDF estimator, and smoothing introduces an additional bias.}
\begin{equation} \label{eq:F_star_hat_definition}
\widehat{F}^{*}\left(v\mid p_n\right)\coloneqq\frac{\sum_{l:N_{l}=n}\sum_{i=1}^{N_{l}^{*}}\mathbbm{1}\left(\widehat{V}_{il}\leq v\right)}{\sum_{l:N_{l}=n}N_{l}^{*}}.
\end{equation}

After computing estimates of the CDF of private costs conditional on entry, we use the copula restriction in \eqref{eq:restrictions_for_theta_and_F}
 to estimate the copula parameter $\theta_0$  and the marginal CDF of private costs $F(\cdot)$. Although the restriction holds at a continuum of points $v\in[\underline v,\overline v]$, we use a finite grid in practice. The boundaries of the support of the distribution of private costs can be estimated using the estimated inverse bidding function together with the maximum and minimum observed bids. Consider a grid $v_1<\ldots<v_J$ within the estimated boundaries of the support.
 We estimate $\theta_0,F(v_1),\ldots,F(v_J)$ by solving the following optimization problem:
 \begin{align*}
     (\widehat\theta,\widehat F(v_1),\ldots,\widehat F(v_J))\coloneqq \underset{\theta,y_{1},...,y_{J}}{\arg\min}\sum_{n\in\mathcal N}\sum_{j=1}^J\Big(Q(\widehat F^*(v_j\mid p_n),\widehat p_n;\theta)-y_j\Big)^2\widehat W(n,j),
 \end{align*}
 subject to the constraints $\theta\in\Theta$ and  $0\leq y_1\leq \ldots\leq y_J\leq1$. Here, $\widehat W(n,j)$ denotes the estimated GMM-efficient weights, which assign zero weight to cross-$(j,n)$ restrictions, as described in the Supplement.\footnote{The asymptotic distribution of $\widehat F^*(v_j\mid p_n)$ is determined by that of the kernel estimators of the density of bids, which are asymptotically independent across different $j$'s and $n$'s.} Note that the optimization problem is quadratic in the $y$'s and, therefore, can be solved in two steps by first concentrating out $F(v_j)$:
 \begin{align*}
     (\widehat{F}(v_1; \theta),\ldots,\widehat{F}(v_J; \theta))\coloneqq 
     \underset{0\leq y_{1}\leq\cdots\leq y_{J}\leq 1}{\arg\min}\sum_{n\in\mathcal N}\sum_{j=1}^J\Big(Q(\widehat F^*(v_j\mid p_n),\widehat p_n;\theta)-y_j\Big)^2\widehat W(n,j),
 \end{align*}
 where $\widehat{F}(v_j; \theta)$ denotes the estimator of the CDF of private costs for a given $\theta$. In the second step,  $\widehat{\theta}$ minimizes
\begin{align*}
\sum_{n\in\mathcal N}\sum_{j=1}^J\Big(Q(\widehat F^*(v_j\mid p_n),\widehat p_n;\theta)-\widehat{F}(v_j; {\theta})\Big)^2\widehat W(n,j).
\end{align*}
The first step is a quadratic programming problem under linear inequality constraints, and the second step is a single-dimensional optimization problem that can be solved using a grid search.
 
 Standard errors can be computed using the usual efficient GMM formulas. In the Supplement, we describe how to use the bootstrap to compute the weights $\widehat W(n,j)$, accounting for the uncertainty in the estimation of the CDF of private costs, including the estimation of the inverse bidding function and the CDF and PDF of the bids.\footnote{Because of the estimation of the inverse bidding strategy in the first step, the GMM-efficient weights depend on the bidding strategy’s derivatives. A plug-in estimator for these derivatives requires an additional smoothing parameter and converges slowly. By contrast, our bootstrap-based approach avoids both of these complications.}

Using \eqref{eq:F=Q}, \eqref{eq:F_star_hat_definition}, and $\widehat{\theta}$, we construct an estimator of the marginal distribution $F(\cdot)$ by averaging over the  values of $n$ in $\mathcal{N}$: \[
\widehat{F}\left(v\right)\coloneqq\frac{1}{\left|\mathcal{N}\right|}\sum_{n\in\mathcal{N}}Q\left(\widehat{F}^{*}\left(v\mid\widehat{p}_{n}\right),\widehat{p}_{n};\widehat{\theta}\right),
\] where $|\mathcal N|$ denotes the number of elements in $\mathcal{N}$. Finally, we construct a plug-in estimator \[
\widehat{\kappa}_{n}\coloneqq\int_{\widehat{\underline{v}}}^{\widehat{\overline{v}}}C_{2}\left(\widehat{F}\left(v\right),\widehat{p}_{n};\widehat{\theta}\right)\left(\left(1-C\left(\widehat{F}\left(v\right),\widehat{p}_{n};\widehat{\theta}\right)\right)^{n-1}-\left(1-\widehat{p}_{n}\right)^{n-1}\right)dv
\]
for the entry cost $\kappa_n$, where $\left(\widehat{\underline{v}},\widehat{\overline{v}}\right)$ are the minimum and maximum of the pseudo private costs, respectively.

\section{Estimation results}\label{sec:empirical}



\subsection{Data}\label{subsec:data}

The \cite{li2009entry} data for TxDoT ``mowing highway right-of-way'' auctions of awarded projects (held between January 2001 and December 2003) include the following information on each auction: the number of potential bidders, submitted bids, the engineer's estimate, the number of items in a contract, and whether it is a local, state, or interstate contract.\footnote{
The data set does not contain any information on call-for-bids in that period of time. The current TxDoT ``contract letting resource hub'' website provides records on auctions of awarded projects and publicly available information for future lettings, but reveals little about the historical frequency of auction failure.} Although the number of potential bidders varies between $3$ and $26$, in many cases the number of auctions and the number of submitted bids are small. 

The number of items in a project varies between $1$ and $7$. As explained in \cite{li2009entry}, these projects usually involve a main task, ``type-II full width mowing'', with additional tasks such as strip mowing, spot mowing, litter pickup and disposal, and sign installation. Projects with a single item are therefore more likely to be comparable, as they presumably involve the same main task. For multi-item projects, the data do not contain information on the type of additional tasks, which may vary across auctions with the same number of items. \cite{li2009entry} also explain that there can be substantial differences between local, state, and interstate jobs. State jobs are auctioned by the state agency, potentially with different requirements for preparing bid proposals. Interstate jobs can be more complicated because of a higher traffic volume. 

To make our sample as homogeneous as possible, we focus only on local projects with one item. 
We further homogenize bids in our sample by the engineer's estimate; thus, bids are fractions of the engineer's estimate. In the Supplement, we show that the engineer's estimate explains the overwhelming portion of bid variation, and once bids are standardized by the estimate, observed project characteristics explain very little of the remaining variation. We exclude values of $n$ with fewer than $30$ submitted bids to ensure that the CDFs of private costs conditional on entry are precisely estimated. Our final sample includes auctions with the number of potential bidders $n=9,10,12,13,14$.

Table \ref{tab:summary} reports the summary statistics for our sample. The average engineer's estimate is between $\$77{,}493$ and $\$113{,}838$, depending on the number of potential bidders. The variation is substantial, with standard deviations ranging from $\$27{,}760$ to $\$48{,}493$. The average submitted bid as a fraction of the engineer's estimate is between $1.0$ and $1.11$, depending on the number of potential bidders. The minimum and maximum bids are $0.7$ and $1.56$, respectively. The average winning bid as a fraction of the estimate ranges from $0.90$ to $1.01$. The percentage of winning bids above the engineer's estimate is between $13.3\%$ and $43.8\%$, depending on the number of potential bidders, confirming that the reserve price is not enforced. The largest winning bid as a fraction of the estimate is $1.25$, and the smallest is $0.70$.

\begin{table}[]
\caption{Summary statistics for the sample of local projects with one item and at least 30 bids for different numbers of potential bidders}
    \label{tab:summary}
    \centering
  \footnotesize
\begin{tabular}{rrrrrr}
  \toprule
  \textbf{Potential bidders} & \textbf{9} & \textbf{10} & \textbf{12} & \textbf{13} & \textbf{14} \\\midrule
    &   &   &   &   &   \\
  Number of auctions & 15 & 15 & 16 & 11 & 10 \\
  Number of bids & 40 & 41 & 43 & 41 & 40 \\
    &   &   &   &   &   \\
  Engineer's estimate (dollars) &   &   &   &   &   \\
  mean & 104,813 & 89,489 & 113,838 & 84,025 & 77,493 \\
  std.dev & 44,333 & 39,547 & 48,493 & 31,496 & 27,760 \\
  std.err & 11,447 & 10,211 & 12,123 & 9,496 & 8,778 \\
    &   &   &   &   &   \\
  Bids (fraction of engineer's estimate) &   &   &   &   &   \\
  mean & 1.068 & 1.004 & 1.106 & 1.037 & 1.057 \\
  std.dev & 0.165 & 0.172 & 0.167 & 0.204 & 0.169 \\
  min & 0.815 & 0.721 & 0.799 & 0.703 & 0.722 \\
  max & 1.445 & 1.471 & 1.470 & 1.556 & 1.530 \\
    &   &   &   &   &   \\
  Winning bids (fraction of engineer's estimate) &   &   &   &   &   \\
  mean  & 0.952 & 0.921 & 1.011 & 0.898 & 0.959 \\
  std.dev  & 0.065 & 0.131 & 0.129 & 0.153 & 0.117 \\
  min  & 0.815 & 0.721 & 0.799 & 0.703 & 0.722 \\
  max  & 1.106 & 1.207 & 1.249 & 1.148 & 1.124 \\
  \% above estimate & 13.3 & 13.3 & 43.8 & 18.2 & 40.0 \\
  std.err for \% above estimate & 8.8 & 8.8 & 12.4 & 11.6 & 15.5 \\
    &   &   &   &   &   \\\bottomrule
\end{tabular}
\end{table}

\subsection{Entry probabilities, signal informativeness, entry costs, and distribution of private costs}\label{subsec:estimation_results_hard}

The estimated probabilities of entry conditional on having at least two active bidders and the implied estimated unconditional equilibrium probabilities of entry $\widehat p_n$ are displayed in Figure \ref{fig:entry_probability}.\footnote{The probability of entry conditional on having at least two active bidders is estimated by the sample analogue of $\E[N^*_{l}/N_{l}\mid  N_l=n]$.} The estimated implied probability of entry varies between $15\%$ and $27\%$, depending on the number of potential bidders.  The difference between the estimated probability of entry conditional on at least two active bidders and the implied unconditional probability $\widehat p_n$ can be substantial and should not be ignored. Both probabilities are non-monotone in the number of potential bidders. Due to the requirement of at least two active bidders, the relationship between the entry probability and the number of potential bidders can be non-monotone even for the same entry cost $\kappa$. This is unlike the case studied in \cite{MSX}, where for the same entry cost, the equilibrium entry probability is monotonically decreasing in the number of potential bidders.

\begin{figure}
    \centering
    \begin{subfigure}{0.45\textwidth}
      \includegraphics[width=\linewidth]{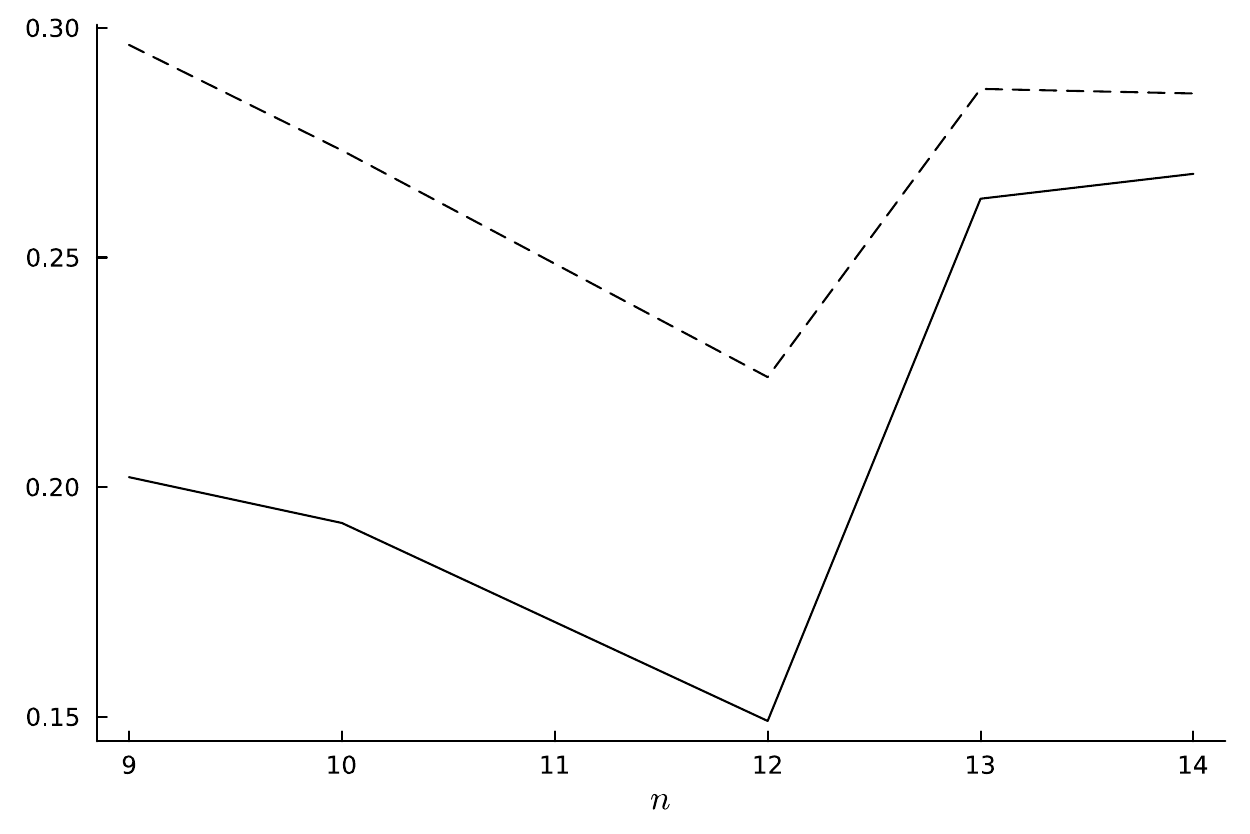}
      \caption{The estimated probabilities of entry conditional on at least two active bidders (dashed line) and the estimated unconditional entry probabilities $p_n$ (solid line) for different numbers of potential bidders $n$ \label{fig:entry_probability}}
    \end{subfigure}
    \hspace{1ex}
    \begin{subfigure}{0.45\textwidth}
      \includegraphics[width=\linewidth]{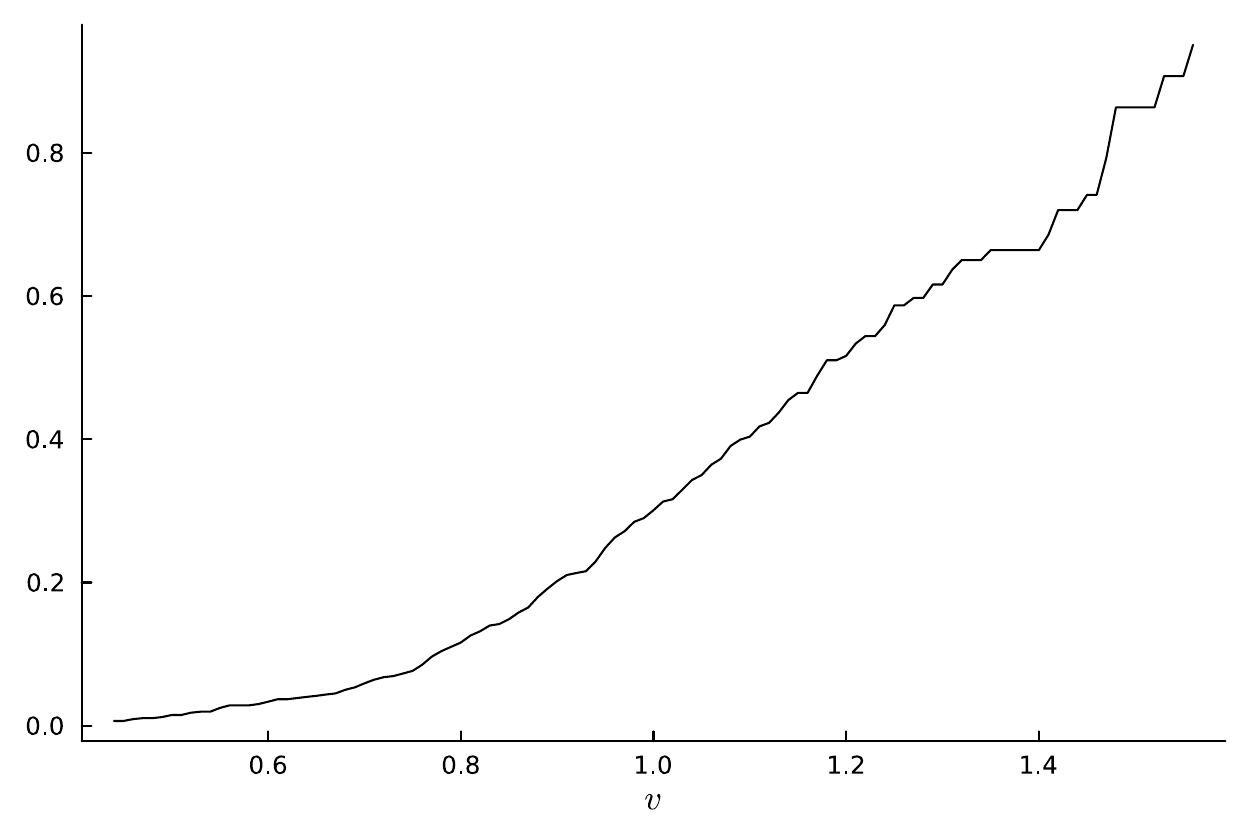}
      \caption{The estimated CDF $F(v)$ of private costs \label{fig:F}\\ \\ \\ }
    \end{subfigure}
    \begin{subfigure}{0.45\textwidth}
      \includegraphics[width=\linewidth]{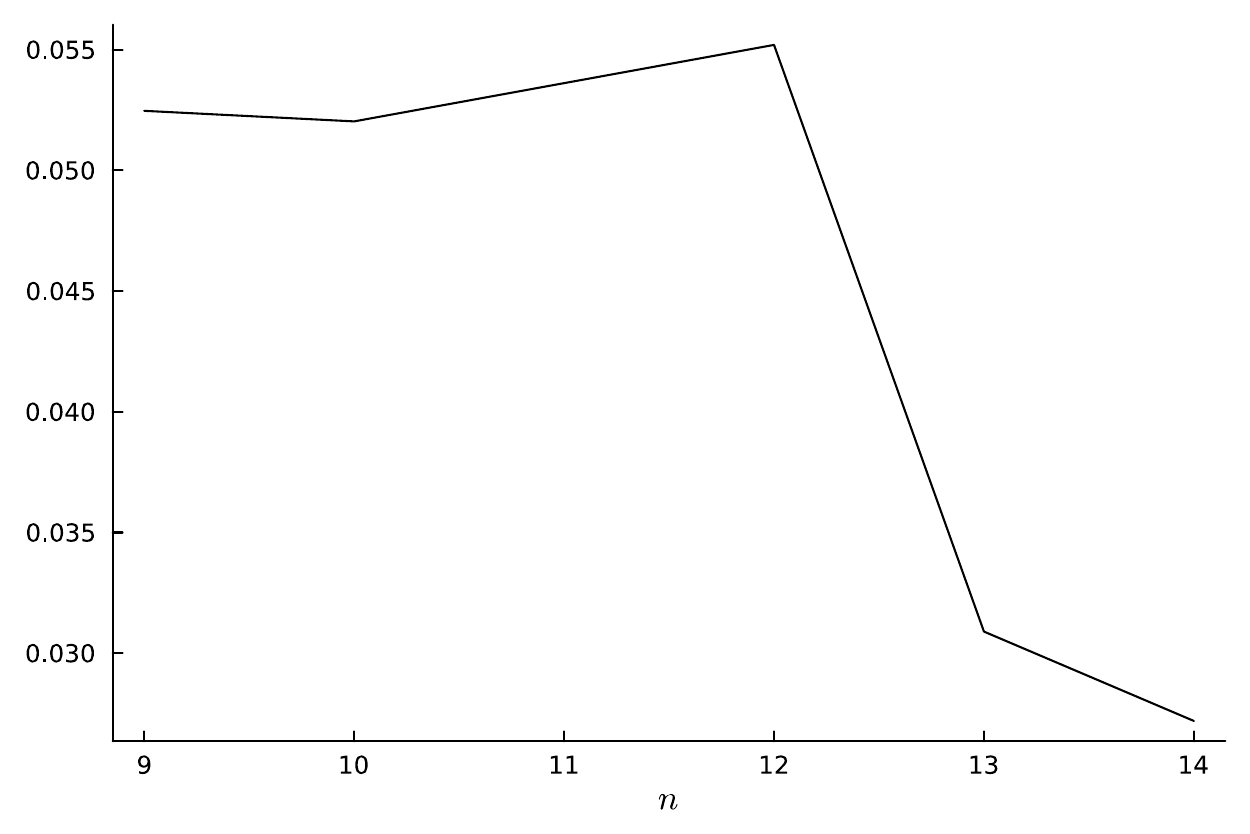}
      \caption{The estimated entry  costs $\kappa_n$ for different numbers of potential bidders $n$ \label{fig:entry_cost}}
    \end{subfigure}
  \caption{Estimates of the entry probability $p_n$, entry cost $\kappa_n$, and  CDF $F(\cdot)$ of private costs \label{fig:estimates}}
  \end{figure}

To estimate the PDF of bids, which is required for estimating the inverse bidding function, we use the triweight kernel in the construction of $\widehat{g}\left(b\mid n\right)$. We follow the rule of thumb for bandwidth selection and set the bandwidth to  $3.15 \cdot\widehat{\sigma}_{b,n}(\sum_{l:N_{l}=n}N_{l}^{*})^{-1/5}$, where $\widehat{\sigma}_{b,n}$ is the sample standard deviation of the bids in auctions with $n$ potential bidders
and $\sum_{l:N_{l}=n}N_{l}^{*}$ is the sample size.\footnote{When computing the confidence intervals, we change the sample size component to $(\sum_{l:N_{l}=n}N_{l}^{*})^{-1/5-\epsilon}$  for minor under-smoothing with $\epsilon=1/17$.} 
After the correction for at least two active bidders, we obtain monotonically increasing estimates of the inverse bidding functions. The estimated support for the distribution of private costs is $[0.47, 1.56]$.

We choose the Frank copula to model the dependence between private costs and signals, and we use the efficient two-step GMM estimator to estimate the copula parameter $\theta$ and the marginal CDF of private costs. We use a grid of private cost values ranging from $0.6$ to $1.5$, with increments of $0.05$, to set up the estimating equations for $\theta$. Table \ref{table:theta_rho} shows the estimates of $\theta$  and the corresponding Spearman rank correlation coefficient $\rho$ obtained using the undersmoothing bandwidth. According to our estimates, the signals are moderately informative with a $\theta$ estimate of $5.54$, which corresponds to Spearman's $\rho$ of $0.68$. We construct the confidence interval for $\rho$ by converting the confidence interval for $\theta$ using the one-to-one relationship between the two parameters.  The resulting 95\% confidence interval for Spearman's $\rho$ is $[0.61, 0.74]$.

Note that the Frank copula becomes the independence copula when $\theta=0$. However, with a $t$-ratio of $11.59$, we can reject independence conclusively. In other words, the data reject the entry model of \cite{levin1994equilibrium}, in which bidders are completely uninformed about their valuations and enter randomly. 
\begin{table}
\caption{The estimates of the copula parameter  and its implied Spearman rank correlation with their standard errors (in parentheses) and 95\% confidence intervals with the undersmoothing bandwidth}
\label{table:theta_rho}
\centering
\footnotesize
\begin{tabular}{rcc}
\toprule
 & \bf{Estimate} & \bf{95\% confidence interval} \\
 \midrule
Copula parameter $\theta$ & $5.54$ & $[4.60, 6.48]$ \\
 & $(0.48)$ & \vspace{0.7ex} \\
Spearman correlation $\rho$ & $0.68$ & $[0.61, 0.74]$ \\
\bottomrule
\end{tabular}
\end{table}

 The estimated CDF of private costs $F(\cdot)$ is shown in Figure \ref{fig:F}. The engineer's estimate ($v=1$) corresponds to the $29$th percentile of the distribution of private costs; that is, the probability of drawing a private cost above the engineer's estimate (the nominal reserve price) is $71\%$. This may help explain why the reserve price was not enforced.
 
 We estimate the entry cost parameter $\kappa_n$  for each value of the number of potential bidders $n$ using equation \eqref{eq:entry_cost_ID}. The estimates are $\widehat\kappa_n=0.053, 0.052, 0.055, 0.031$, and $0.027$  for $n=9,10,12,13$, and $14$, respectively. The results are plotted in Figure \ref{fig:entry_cost}. For example, $\widehat{\kappa}_9=0.053$ implies that the entry cost in markets with $9$ potential bidders is estimated to be $5.3\%$ of the engineer's estimate. In auctions with $9$--$12$ potential bidders, the entry costs are at least $5.2\%$ of the engineer's estimate: $\$5{,}499$, $\$4{,}656$, and $\$6{,}283$ for $n=9$, $10$, and $12$, respectively. The entry cost is lower for auctions with $13$--$14$ potential bidders: $3.1\%$ and $2.7\%$, or $\$2{,}596$ and $\$2{,}107$ for $n=13$ and $14$, respectively. The entry cost decreases sharply for $n\geq13$. The negative association between the entry cost and the number of potential bidders may explain the variation in $n$ across auctions. These findings are intuitive: projects whose costs are easier to assess are expected to attract more potential bidders.

For robustness, we also estimate the model using the Joe copula as an alternative specification. Under the Joe copula, the estimate of Spearman's $\rho$ is $0.6$, which is close to the Frank copula estimate of $0.68$. The remaining estimation results and counterfactual conclusions are similarly robust to this alternative. Details are provided in Section S2.2 of the online Supplement.

\section{Format comparisons}\label{sec:counterfactuals}

In this section, we discuss the effect of changing the format from bid requirements to reserve prices on the expected winning bid and probability of auction failure. Recall that under the bid requirement format, an auction fails when fewer than two bidders enter. Under the reserve price format, an auction fails when there are no active bidders, that is, when all potential bidders either drew signals above the entry threshold or had private costs above the reserve price.

\subsection{Fixed signal informativeness}\label{sec:counterfac}

First, we conduct the counterfactual calculations using the estimated level of signal informativeness $\widehat\theta=5.54$ or $\widehat\rho=0.68$. For the reserve price format, we set the reserve price $r=1.0$ (the nominal reserve price in TxDoT auctions). Recall that according to our estimates, this reserve price can be perceived as aggressive: the probability of drawing a private cost below it is only $29\%$.\footnote{In this model with endogenous entry, to derive the optimal reserve price, one has to take a stance about the expected procurement cost when an auction fails. After specifying the expected procurement cost in the event of failure, one can derive the optimal reserve price with respect to the unconditional cost of procurement. However, specifying the expected procurement cost in the event of failure requires additional assumptions that may be difficult to justify.} Initially, we keep the entry costs $\kappa_n$ at their estimated levels for each $n$.  The results are reported in Table \ref{tab:counterfact}.

\begin{table}[]
\caption{Counterfactual entry and bidding probabilities, auction failure probabilities, and expected winning bids under the two formats for different numbers of potential bidders $n$ at the estimated entry costs $\kappa_n$ }
    \label{tab:counterfact}
    \centering
  \footnotesize
  \renewcommand{\arraystretch}{1.1} 
\begin{tabular}{>{\centering\arraybackslash}>{\centering\arraybackslash}p{1cm}>{\centering\arraybackslash}p{1.4cm}>{\centering\arraybackslash}p{1.4cm}>{\centering\arraybackslash}p{1.4cm}>{\centering\arraybackslash}p{.1cm}>{\centering\arraybackslash}p{1.4cm}>{\centering\arraybackslash}p{1.4cm}>{\centering\arraybackslash}p{1.4cm}>{\centering\arraybackslash}p{1.4cm}}
  \toprule
  & \multicolumn{3}{c}{\textbf{Bid requirement}} & & \multicolumn{4}{c}{\textbf{Reserve price}} \\
  \cline{2-4}\cline{6-9}
  $n$ & {prob. entry}& {prob. failure} & {expect. win. bid } &  & {prob. entry}& {prob. bidding} & {prob. failure} & {expect. win. bid }  \\ \midrule
 9 & 0.186 & 0.48 & 0.958 &  & 0.201 & 0.139 & 0.259 & 0.921 \\
  10 & 0.179 & 0.442 & 0.948 &  & 0.190 & 0.133 & 0.239 & 0.915 \\
  12 & 0.134 & 0.508 & 0.938 &  & 0.162 & 0.116 & 0.227 & 0.908 \\
  13 & 0.257 & 0.116 & 0.897 &  & 0.230 & 0.156 & 0.111 & 0.874 \\
  14 & 0.263 & 0.084 & 0.882 &  & 0.237 & 0.159 & 0.088 & 0.862 \\\bottomrule
\end{tabular}
\end{table}

Under the bid requirement format, all entering bidders bid, so the probabilities of entry and bidding coincide. Under the reserve price format, only entering bidders with private costs below the reserve price bid. The entry probability is lower under the bid requirement format in auctions with $9$--$12$ potential bidders and higher in auctions with $13$--$14$ potential bidders. Nevertheless, the bidding probability is always higher under the bid requirement format, and the differences can be substantial. For example, in auctions with $13$ potential bidders, the difference in bidding probabilities between the two formats is $10.1$ percentage points. Despite this, the auction failure probability is substantially higher under the bid requirement format for $n$ between $9$ and $12$, and very similar for $n$ equal to $13$ or $14$. In the former case, the auction failure probability under the bid requirement format exceeds that under the reserve price format by $22.1$, $20.3$, and $28.1$ percentage points for $n=9$, $10$, and $12$, respectively. 

The substantial differences in auction failure probabilities are driven by the probability of exactly one bidder entering, in which case the auction is canceled under the bid requirement format: $32\%$, $30\%$, and $33\%$ for $n=9$, $10$, and $12$, respectively. That probability is substantially lower for $n=13$ and $14$: $9\%$ and $7\%$, respectively. To understand this difference between auctions with $9$--$12$ and those with $13$--$14$ potential bidders, note that the latter have substantially higher entry probabilities, which can in turn be explained by lower entry costs. Recall that the entry costs are over $5.2\%$ for $n$ between $9$ and $12$ and under $3.1\%$ for $n$ equal to $13$ or $14$.

Finally, the results in Table \ref{tab:counterfact} show that the expected winning bid is lower under the reserve price format by $2.0\%$--$3.7\%$ of the engineer's estimate, depending on the number of potential bidders. Therefore, in the TxDoT case, switching to the reserve price format leads not only to lower (or similar) auction failure probabilities but also to a reduction in the expected cost of procurement.  

Next, to eliminate the impact of varying entry costs, we consider the same counterfactual outcomes while setting the entry cost to its estimated weighted average across $n$: $0.046$, or $4.6\%$ of the engineer's estimate. We also extend the number of potential bidders $n$ to $4$--$25$ to consider small and large markets. The results are shown in Figure \ref{fig:counter}.

\begin{figure}
  \centering
  \begin{subfigure}{0.45\textwidth}
    \includegraphics[width=\linewidth]{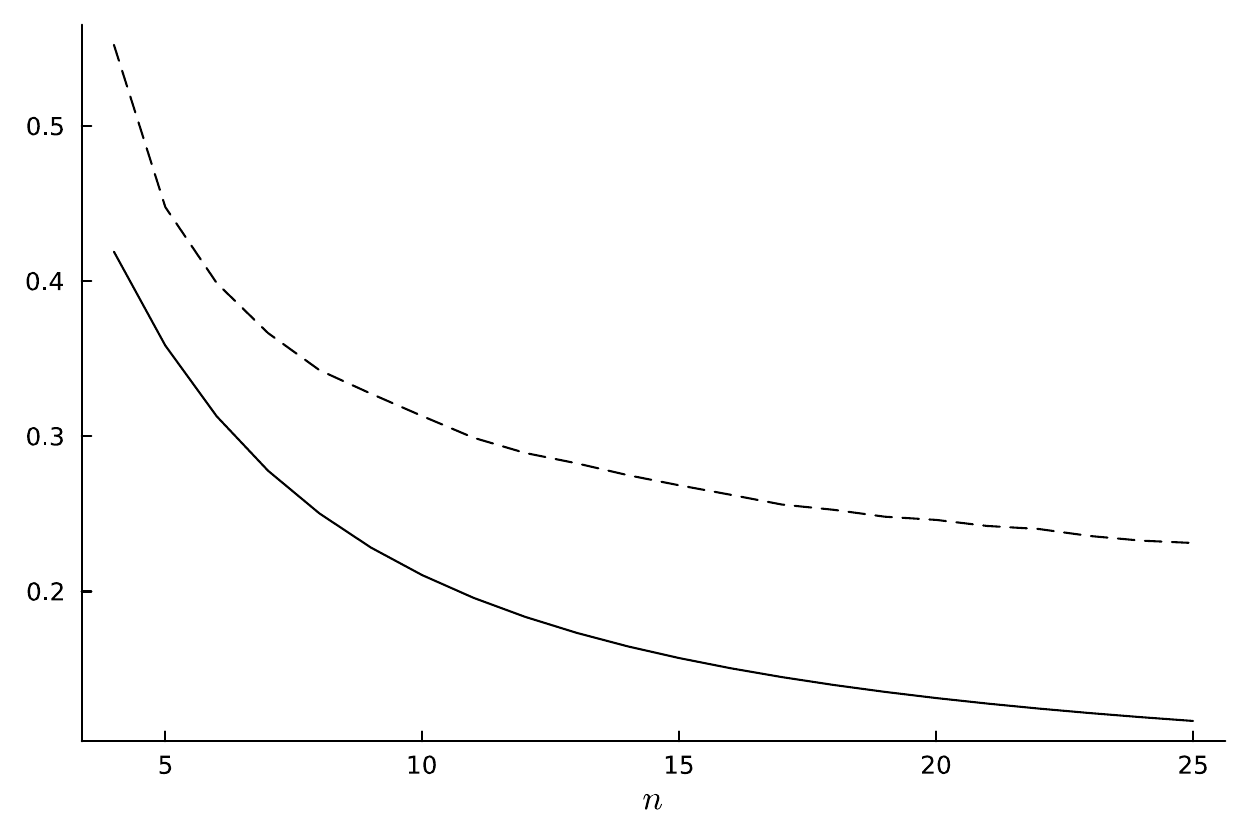}
    \caption{Auction failure probability \label{fig:counter_auction_failure} }
  \end{subfigure}
  \hspace{1ex}
  \begin{subfigure}{0.45\textwidth}
    \includegraphics[width=\linewidth]{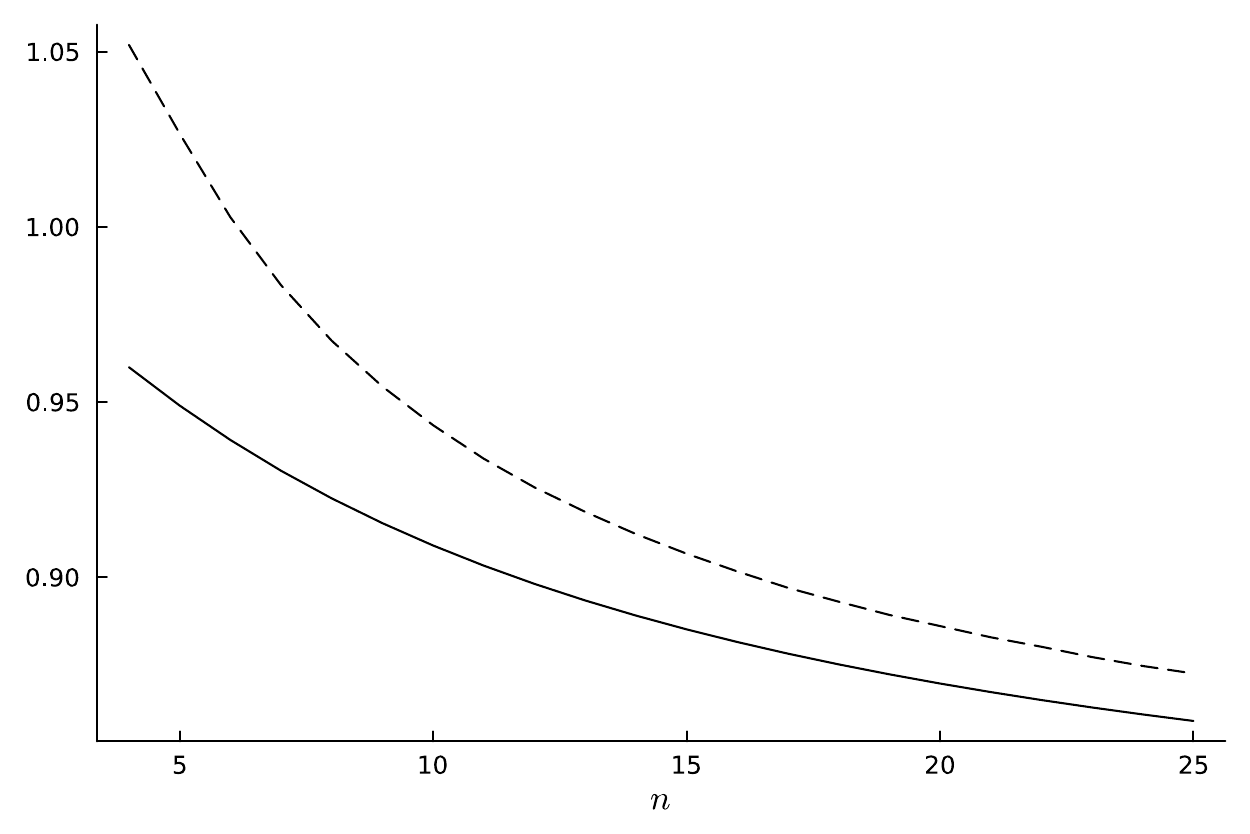}
    \caption{Expected winning bid \label{fig:counter_winning_bid}}
  \end{subfigure}
\caption{Counterfactual auction failure probabilities and expected winning bids for different numbers of potential bidders $n$ and estimated weighted average entry cost $\kappa=0.046$ under the bid requirement (dashed line) and reserve price (solid line) formats; the reserve price $r=1$\label{fig:counter}}
\end{figure}

For each number of potential bidders $n$, the reserve price format results in an auction failure probability that is $8.6$--$13.3$ percentage points lower than under the bid requirement format. In larger markets with $n\geq 14$, the difference is over $11$ percentage points in favor of the reserve price format. Similarly, the reserve price format dominates in terms of the expected winning bid by $1.4\%$--$9.2\%$ of the engineer's estimate, with larger differences for smaller $n$. 

We conclude that, at the estimated level of signal informativeness $\rho=0.68$, the reserve price format is preferred in terms of auction failure probability and expected winning bid, even at the aggressively set reserve price corresponding to the $29$th percentile of the distribution of private costs. This conclusion holds as long as entry costs are sufficiently high (over $3.1\%$ of the engineer's estimate). Such entry cost levels result in a low entry probability and, consequently, a high probability of exactly one active bidder, in which case the auction is canceled under the bid requirement format. This holds even though, under the reserve price format, a non-negligible fraction of entering bidders have private costs above an aggressively set reserve price. For the bid requirement format to be preferred, entry costs must be sufficiently low to substantially reduce the probability of exactly one bidder entering.

\subsection{Varying signal informativeness}\label{subsec:varying_info}

Next, we compare the procurement outcomes for the bid requirement and reserve price formats in different information environments (i.e., for different levels of signal informativeness as measured by Spearman's $\rho$). We focus on auctions with ten potential bidders; the results for other values of $n$ are qualitatively similar. In the calculations below, we set the entry cost to its estimated level for $n=10$: $5.2\%$ of the engineer's estimate.

\begin{figure}
    \centering
    \begin{subfigure}{0.50\textwidth}
    \includegraphics[width=\linewidth]{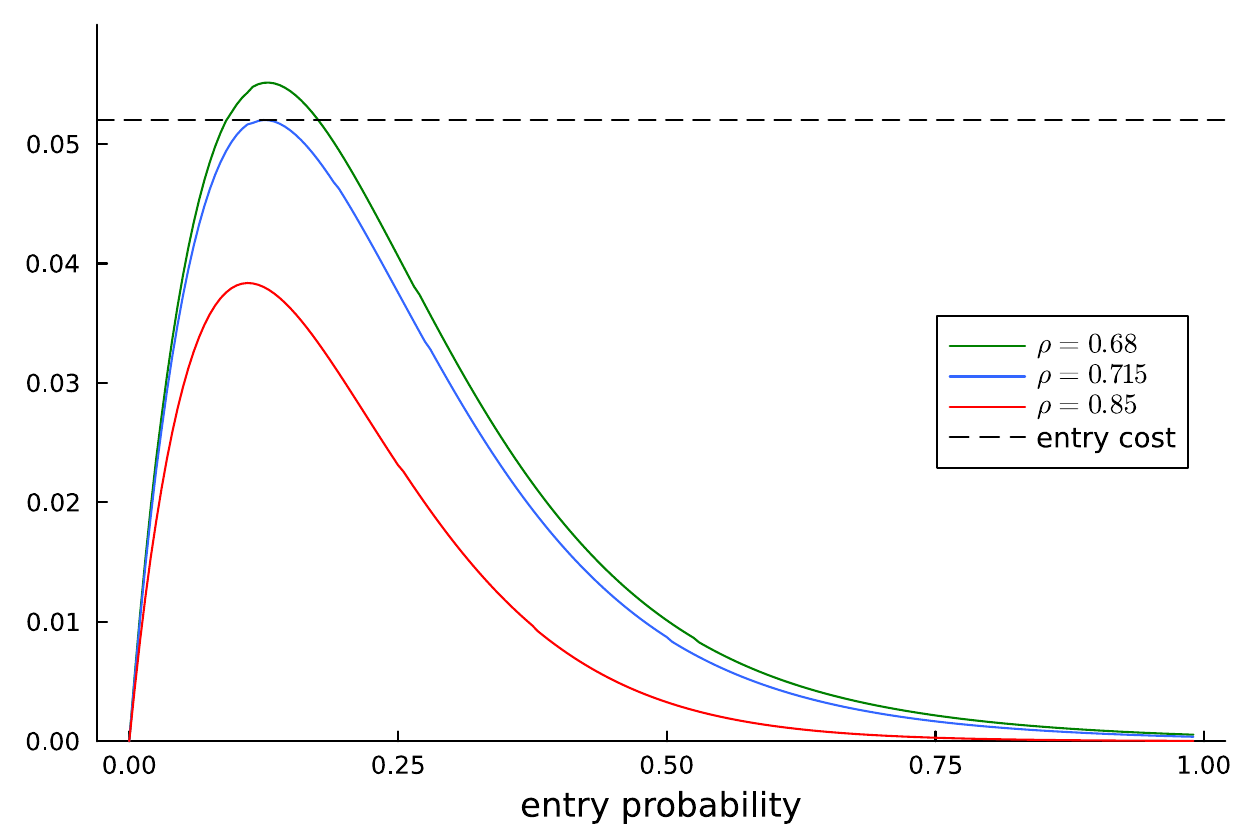}
    \caption{Bid requirement \label{fig:marginal_at_least_2} }
  \end{subfigure}
  \vspace{1em}
  \begin{subfigure}{0.50\textwidth}
    \includegraphics[width=\linewidth]{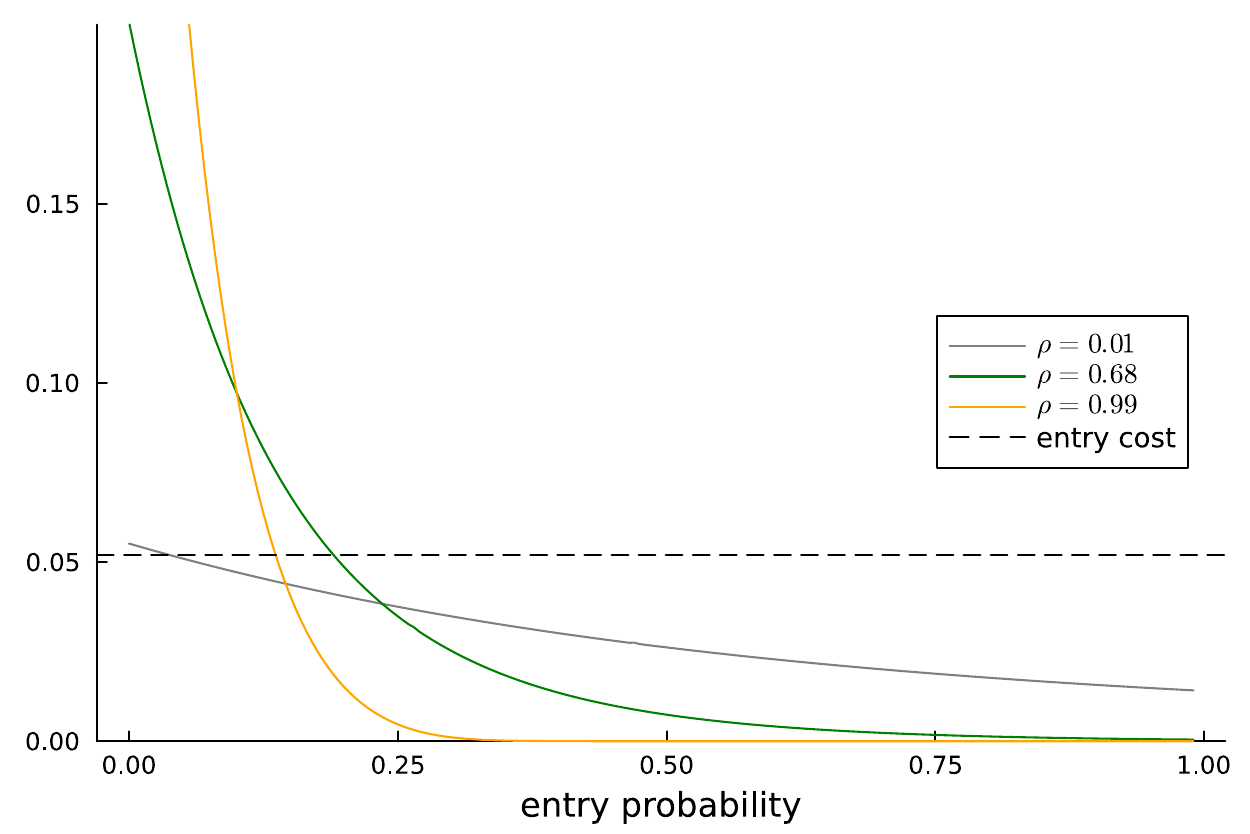}
    \caption{Reserve price \label{fig:marginal_reserve}}
  \end{subfigure}
\caption{Entry cost and marginal entrant's expected revenue from entry under the bid requirement and reserve price formats in auctions with $n=10$ potential bidders for different entry probabilities and levels of signal informativeness as measured by Spearman's rank correlation $\rho$}
    \label{fig:marginal_10_entry}
\end{figure}

For the bid requirement format, Figure \ref{fig:marginal_at_least_2} plots the marginal entrant's expected revenue from entry for different entry probabilities and levels of signal informativeness (measured by Spearman's rank correlation $\rho$ between private costs and signals). The figure shows that, for all entry probabilities, the marginal entrant's expected revenue decreases with $\rho$; that is, the expected revenue is lower under more informative signals. Consequently, the stable (right) equilibrium entry probability decreases with signal informativeness. Moreover, above a sufficiently high level of informativeness ($\rho=0.715$ in this case), entry ceases completely, resulting in certain auction failure under the bid requirement format. Figure \ref{fig:marginal_reserve} plots the same objects under the reserve price format. The entry probability is strictly positive for all levels of signal informativeness; however, it is non-monotone in $\rho$. Moreover, for sufficiently high $\rho$, the reserve price format can support much larger entry costs than the bid requirement format: one would observe strictly positive entry probabilities under the reserve price format but no entry under the bid requirement format.

\begin{figure}
    \centering
    \begin{subfigure}{0.51\textwidth}
    \includegraphics[width=\linewidth]{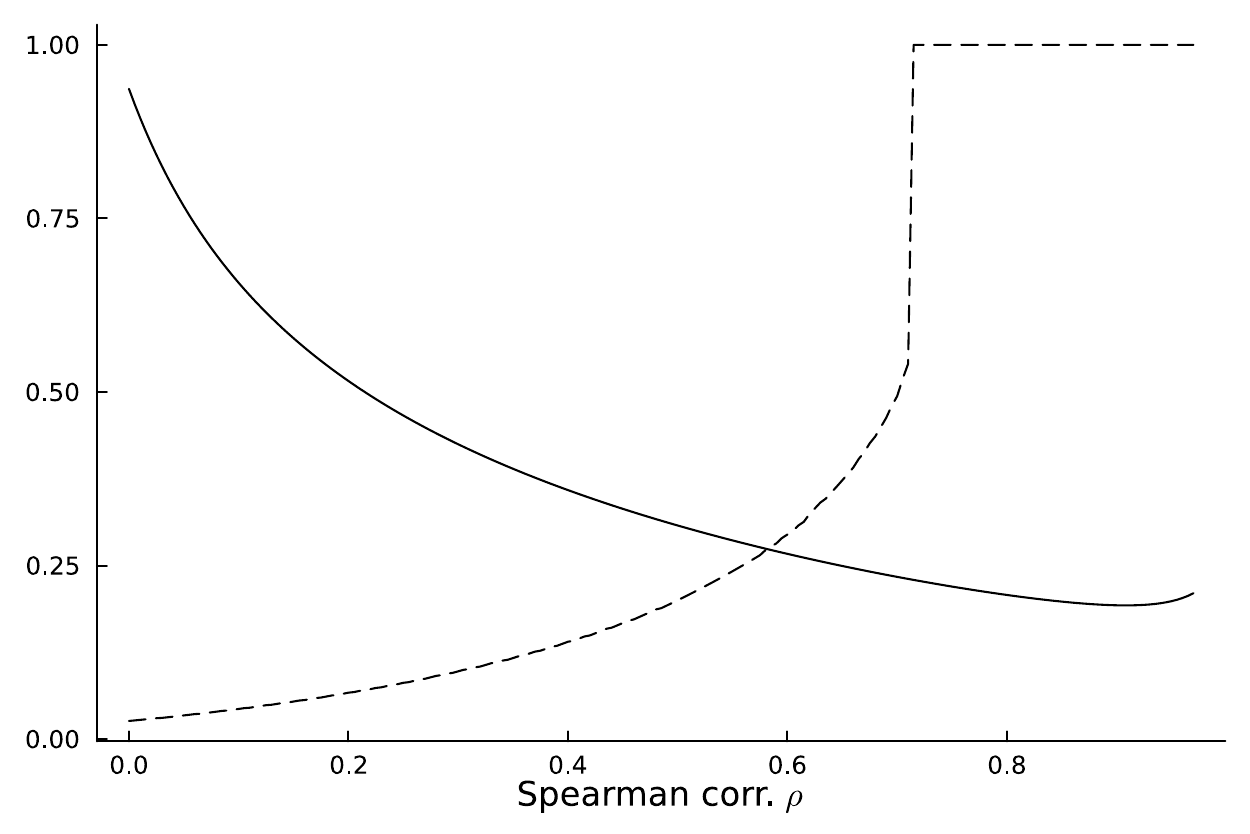}
    \caption{Auction failure probability \label{fig:failure_10} }
  \end{subfigure}
  \vspace{1em}
  \begin{subfigure}{0.51\textwidth}
    \includegraphics[width=\linewidth]{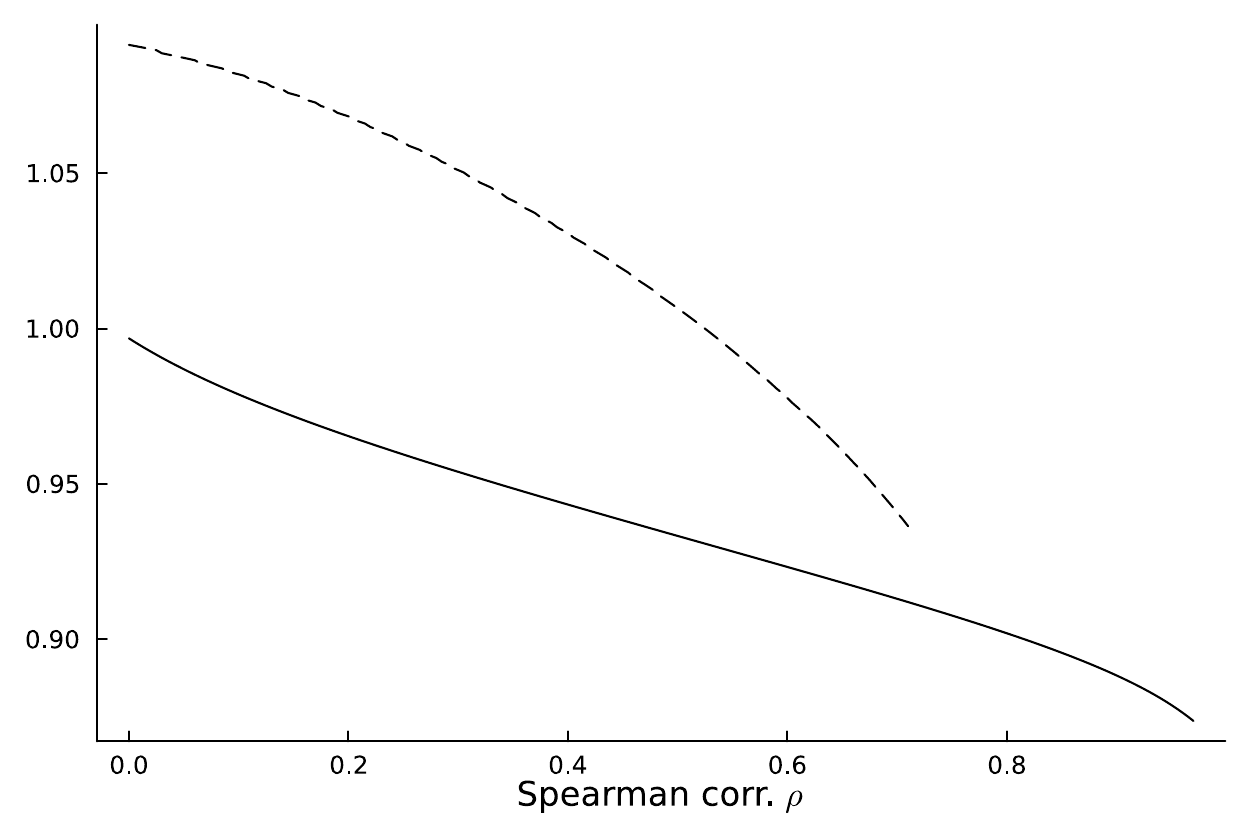}
    \caption{Expected winning bid \label{fig:win_bid_10}}
  \end{subfigure}
\caption{Probabilities of auction failure and expected winning bids under the bid requirement (dashed line) and reserve price (solid line) formats in auctions with $n=10$ potential bidders for different levels of signal informativeness as measured by Spearman's rank correlation $\rho$}
    \label{fig:failure_win_bid_10}
\end{figure}

Figure \ref{fig:failure_win_bid_10} reports the auction failure probabilities and expected winning bids under the two formats for different levels of signal informativeness $\rho$. Under the bid requirement format, there is no entry if signal informativeness $\rho$ exceeds $0.715$, and therefore the auction fails with probability one, confirming the prediction of Proposition \ref{prop:limit_theta_infty}\ref{prop:limit_theta_bid_req}. In such cases, the expected winning bid is undefined and is not reported in Figure \ref{fig:win_bid_10}. Under the reserve price format, the auction failure probability is strictly below one for all levels of signal informativeness, consistent with Proposition \ref{prop:limit_theta_infty}\ref{prop:limit_theta_res_price}. In particular, for $\rho>0.58$, the reserve price format yields a lower auction failure probability. Moreover, the reserve price format results in a lower expected cost of procurement (winning bid) at all levels of signal informativeness.

Recall that under both formats, the auction failure probability is monotonically decreasing in the entry probability. One of the stark differences between the two formats is that the entry probability is decreasing in signal informativeness under the bid requirement format, while under the reserve price format it is non-monotone, increasing over most of the range. In particular, the entry probability under the reserve price format peaks at highly (but not perfectly) informative signals with $\rho=0.91$. This non-monotonicity arises because the marginal entrant's expected revenue varies non-monotonically with signal informativeness, as is also evident from the curves in Figure \ref{fig:marginal_reserve}.

Lastly, recall that for these counterfactuals, the reserve price is set at the $29$th percentile of the distribution of private costs. With a less restrictive reserve price $r=1.202$ (which corresponds to the median private cost), the auction failure probability under the reserve price format is between $8\%$ and $14\%$. Moreover, for $\rho > 0.315$, it falls below that under the bid requirement format. Although a higher reserve price increases the expected winning bid, the reserve price format nevertheless dominates the bid requirement format in this metric for $\rho<0.565$.

We conclude that the reserve price format can substantially outperform the bid requirement format when the entry cost is sufficiently high or signals are sufficiently informative. In particular, when signals are highly informative, the bid requirement format leads to auction failure with probability one. The reserve price format is superior in this regard because it can support a broader range of signal informativeness levels and entry costs while maintaining strictly positive bidding probabilities. 

\section{Soft bid requirement format}\label{sec:soft_bid_req}

In this section, we study a variation of the bid requirement format, which we refer to as the soft bid requirement format (``soft format'').\footnote{This extension, along with the ``soft'' and ``hard'' terminology, was suggested by an anonymous referee.} In this format, the bid requirement is relaxed: if only one bidder enters, the government may still accept the bid, provided it is not excessive. Specifically, following \citet{li2009entry}, we assume that the government enters as another bidder with a randomly drawn private cost. The single active bidder's bid is deemed excessive if it exceeds the government's bid.\footnote{The government effectively imposes a random reserve price drawn from the equilibrium bid distribution in the event of only one active bidder.} If the bid is not excessive, the government accepts it; otherwise, the auction fails. Henceforth, the bid requirement format studied in the previous sections is referred to as the hard bid requirement format (``hard format''), and the auction model under the hard (soft) format is referred to as the hard (soft) model. We reuse the same notation as in the previous sections, but with new definitions.

\subsection{Model and equilibrium characterization}

Under the soft format, we assume that if only one bidder enters, the government draws a private cost from the distribution of private costs conditional on entry $F^*(\cdot\mid p_n)$, where $p_n$ is the equilibrium entry probability with $n$ potential bidders. The distribution $F^*(\cdot\mid p_n)$ is a natural reference point for assessing whether the submitted bid is excessive. Moreover, this assumption allows for tractable equilibrium characterization while capturing the idea that the government has some knowledge of the distribution of entering bidders' private costs.

Under this specification, in a symmetric equilibrium, the probability that an active bidder with private cost $v$ wins the auction when the entry probability is $p$, and the expected profit of a potential bidder with signal $s$, are respectively given by
\begin{align*}
  H(v\mid p,n) &= \Lambda^{n-1}\left(v\mid p\right)-\left(1-p\right)^{n-1}F^{*}\left(v\mid p\right),\\
  \Pi(p,n,\kappa,s) &=R(p,n,s)-\kappa,
\end{align*}
where 
\[R(p,n,s)\coloneqq \int_{\underline{v}}^{\overline{v}}C_{2}\left(F\left(v\right),s\right)\left\{ \Lambda^{n-1}\left(v\mid p\right)-\frac{\left(1-p\right)^{n-1}}{p}\left(1-\Lambda\left(v\mid p\right)\right)\right\} dv\]
denotes the expected revenue from entry for a potential bidder with signal $s$ in an auction with $n$ potential bidders and entry probability $p$. As previously, $\Lambda(v\mid p)=1-C(F(v),p)$. The equilibrium entry probability $p_n$ is determined by $\Pi(p_n,n,\kappa_n,p_n)=0$. Unlike under the hard format, the expected revenue of the marginal entrant ($s=p$) is strictly positive at $p=0$. Provided that the dependence between private costs and signals is not perfect,
\[\lim_{p\downarrow 0}R(p,n,p)=\int_{\underline{v}}^{\overline{v}}C_{2}\left(F\left(v\right),0\right)\left(1-C_{2}\left(F\left(v\right),0\right)\right)dv>0.\]
This implies that, unless the entry cost is sufficiently high, zero entry is not necessarily an equilibrium under the soft format. The result is intuitive: even when the entry probability is very small and the likelihood of being the only active bidder is high, there is still a positive probability of winning against the government. 

\begin{figure}
  \centering
  \includegraphics[width=0.5\textwidth]{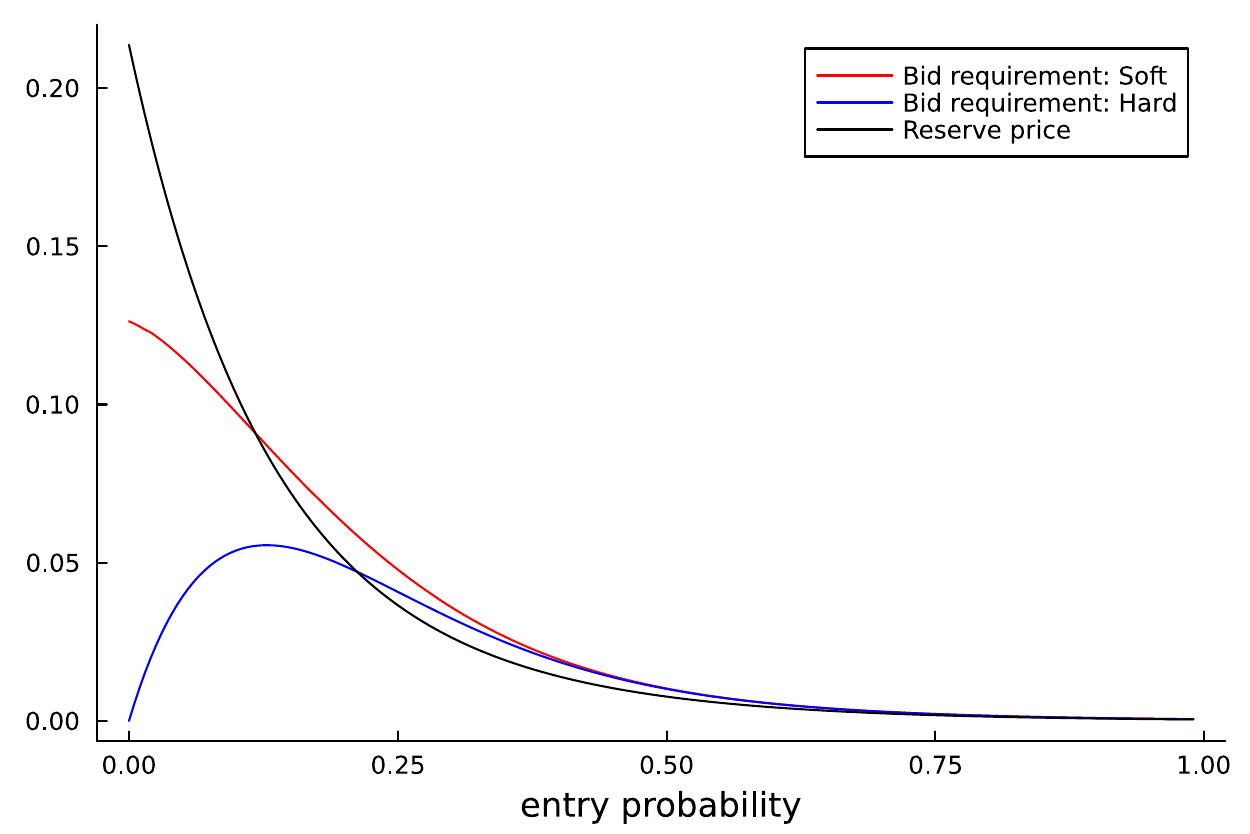}
  \caption{Marginal entrant's expected revenue comparisons under the soft and hard bid requirement and reserve price formats estimated
using the TxDoT data for auctions with 14 potential bidders}
  \label{fig:marginal_bidder_revenue_comparison_all_3}
\end{figure}

Figure \ref{fig:marginal_bidder_revenue_comparison_all_3} compares the marginal entrant's expected revenue from entry under the soft and hard bid requirement and reserve price formats, estimated using the TxDoT data for auctions with $14$ potential bidders. The revenue under the reserve price format is calculated for the reserve price $r=1$.\footnote{The estimates were obtained assuming that data were generated by the soft model; see the estimation details below.} The figure illustrates that the expected revenue under the soft bid requirement is strictly positive for entry probabilities close to zero, unlike in the hard format. Moreover, under the soft format, the expected revenue of the marginal entrant is a decreasing function of the entry probability, implying a unique equilibrium in this application. In general, a non-monotone relationship between the marginal entrant's expected revenue and the entry probability under the soft format cannot be ruled out. However, for sufficiently large $n$, a result similar to Proposition \ref{p:quasi-concavity} can be established for the soft model, implying uniqueness of the stable equilibrium when it exists. 

Figure \ref{fig:marginal_bidder_revenue_comparison_all_3} also illustrates that the soft format can disadvantage marginal entrants relative to the reserve price format, although to a lesser extent than the hard format. When the entry probability is low and signal informativeness is high, the expected revenue of the marginal entrant is substantially lower under the soft format than under the reserve price format. Under the soft format, there is a positive probability of facing competition from the government when the entrant is the only active bidder. 
Furthermore, the result of Proposition \ref{prop:limit_theta_infty}\ref{prop:limit_theta_bid_req} continues to hold under the soft format: with sufficiently informative signals, there is no entry in equilibrium, regardless of the entry cost.\footnote{The proofs of this result and \eqref{eq:expected winning bid soft} are relegated to the online Supplement.}

The probability of the auction succeeding (i.e., having two active bidders or one active bidder with a private cost below the government's draw) is given by
\begin{equation*}
  P(p,n) \coloneqq  1- \left(1-p\right)^{n}-\frac{1}{2}\cdot np\left(1-p\right)^{n-1},
\end{equation*}
where the second term follows from the fact that, when there is only one active bidder, the government and the active bidder draw private costs from the same distribution. Thus, the probability of auction failure under the soft format is $1-P(p,n)$. The expected winning bid under the soft format is given by
\begin{multline} \label{eq:expected winning bid soft}
  \frac{1}{P\left(p,n\right)}\biggl\{ \underline{v}-\left(1-p\right)^{n-1}\overline{v}\left(1+\left(\frac{n}{2}-1\right)p\right)\\
  +\,n\int_{\underline{v}}^{\overline{v}}\Lambda^{n-1}\left(v\mid p\right)\left(1-\frac{n-1}{n}\Lambda\left(v\mid p\right)\right)dv-\frac{n(1-p)^{n-1}}{2p}\int_{\underline{v}}^{\overline{v}}(\Lambda\left(v\mid p\right)-1)^{2}dv\biggr\}.
\end{multline}

Following arguments similar to those in the proof of Proposition \ref{p:inverse_bidding_function}, the inverse bidding function under the soft format is given by
\begin{align}
  \xi(b\mid p_n,n) & = b-\frac{{\eta}_{n}\left(p_n,G\left(b\mid n\right)\right)}{\left(n-1\right)g\left(b\mid n\right)},\text{ where } \label{eq:inverse_bidding_soft}\\
  {\eta}_{n}\left(p,y\right)&\coloneqq\frac{\left(1-p\cdot y\right)^{n-1}-\left(1-p\right)^{n-1}y}{p\left(1-p\cdot y\right)^{n-2}+\left(1-p\right)^{n-1}/(n-1)}, \label{eq:eta_n_soft}
\end{align}
with $G(\cdot\mid n)$ and $g(\cdot\mid n)$ being the CDF and PDF of bids, respectively.\footnote{Under the soft format, when each competitor participates with probability $p$ and submits a bid distributed according to $G(\cdot \mid n)$ in the bidding stage, the probability of winning with bid $b$ is $\left\{ 1-p\cdot G\left(b\mid n\right)\right\} ^{n-1}-\left(1-p\right)^{n-1}G\left(b\mid n\right)$. } These expressions can be used to estimate the distribution of private costs conditional on entry, $F^*(\cdot\mid p_n)$. Assuming a parametric copula family with $C(\cdot,\cdot)=C(\cdot,\cdot;\theta_0)$, the restriction \eqref{eq:restrictions_for_theta_and_F} continues to hold under the soft bid requirement format. Therefore, after estimating $F^*(\cdot\mid p_n)$, one can estimate the copula parameter $\theta_0$ and the marginal distribution of private costs $F(\cdot)$ using the same procedure as in Section \ref{sec:estimation}. The results are reported in the next section.

\subsection{Estimation results}

We estimate the soft model using the same data as in Section \ref{subsec:data} and the GMM-based procedure described in Section \ref{sec:estimation}. However, private costs are estimated using the soft model's inverse bidding function in equations \eqref{eq:inverse_bidding_soft} and \eqref{eq:eta_n_soft}. The asymptotic properties of these estimators are straightforward extensions of those for the hard bid requirement model. For the estimation of the inverse bidding function, we set the bandwidth parameters as described in Section \ref{subsec:estimation_results_hard}. The equilibrium entry probabilities $p_n$ can still be estimated using equation \eqref{eq:p_N_ID} as the data set only contains auctions with at least two active bidders. We continue to use the Frank copula family to specify the joint distribution of private costs and signals, and the same grid of private cost values to estimate the copula parameter $\theta$.

The estimated copula parameter in the soft model is very similar to that in the hard model: $\widehat\theta=5.55$ (versus $5.54$ under the hard model) with a standard error of $0.51$. The corresponding Spearman rank correlation coefficient is $\widehat\rho=0.68$, essentially the same as in the hard model, with a $95\%$ confidence interval of $[0.61,0.74]$. The estimates of the marginal distribution of private costs are also very similar to those in the hard model and are not reported here for brevity.

\begin{figure}
    \centering
    \includegraphics[width=0.5\textwidth]{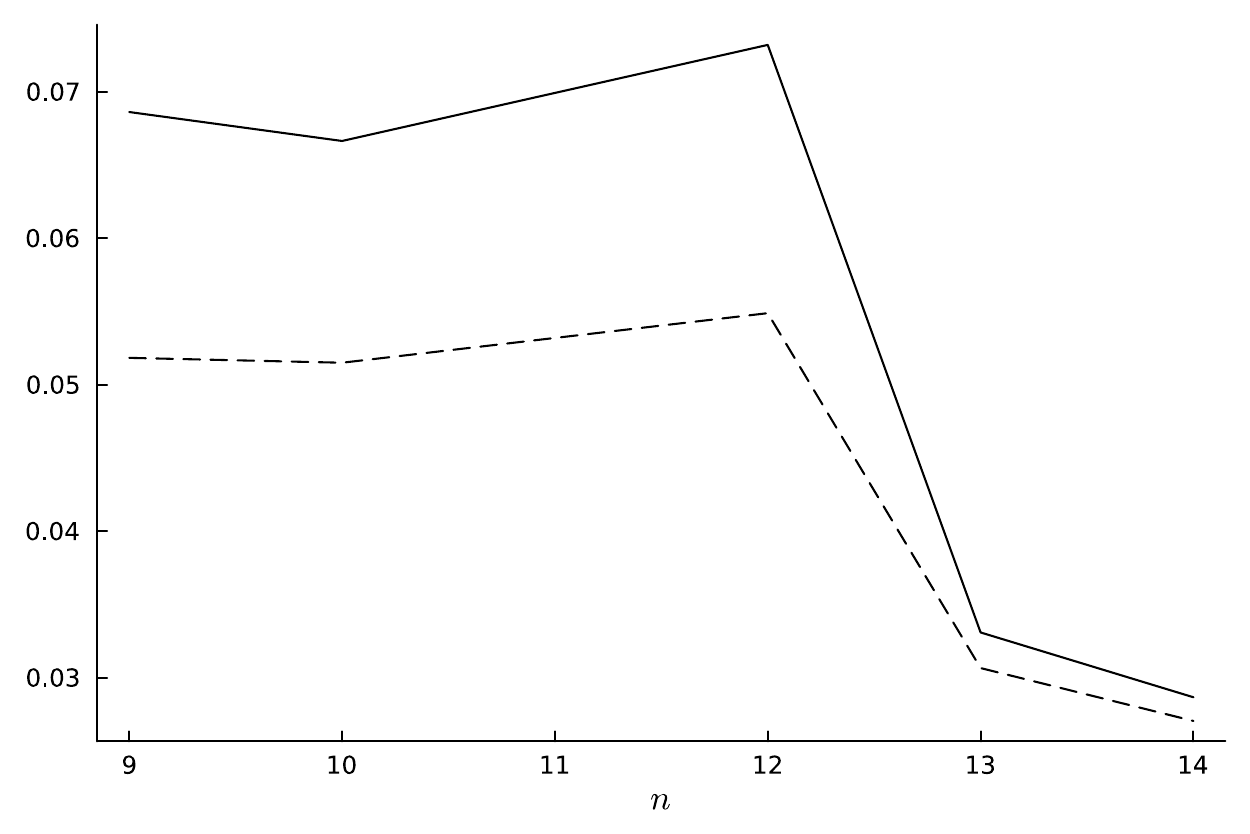}
    \caption{Estimated entry cost $\kappa_n$ under the soft (solid line) and hard (dashed line) bid requirement models for different numbers of potential bidders $n$}
    \label{fig:entry_costs_soft_hard}
\end{figure}

Figure \ref{fig:entry_costs_soft_hard} shows the estimated entry cost $\kappa_n$ under the soft model assumption, with the corresponding estimates under the hard model assumption reported for comparison. The estimated entry cost is higher under the soft model, particularly for $n=9$, $10$, and $12$ (by $0.017$, $0.015$, and $0.018$, respectively). The higher entry costs reflect the fact that the soft model yields higher expected revenues for the marginal entrant than the hard model. To match the resulting equilibrium entry probabilities, the entry cost must be higher under the soft model.

\subsection{Format comparisons}

Since the parameter estimates under the soft model closely resemble those under the hard model, the comparison between the hard bid requirement and reserve price formats based on the soft model estimates aligns with the findings in Section \ref{sec:counterfactuals}. In this section, we compare auction outcomes under the soft bid requirement and reserve price formats in terms of auction failure probabilities and expected winning bids. First, we consider the outcomes using the estimated entry costs $\kappa_n$ for each number of potential bidders $n$. The results are reported in Table \ref{tab:counterfact_soft}.

\begin{table}[]
\caption{Auction failure probabilities and expected winning bids under the soft bid requirement and reserve price formats for different numbers of potential bidders $n$ and reserve price $r$ at the estimated entry costs $\kappa_n$}
    \label{tab:counterfact_soft}
    \centering
  \footnotesize
  \renewcommand{\arraystretch}{1.1}
\begin{tabular}{>{\centering\arraybackslash}>{\centering\arraybackslash}p{1cm}>{\centering\arraybackslash}p{1.4cm}>{\centering\arraybackslash}p{1.4cm}>{\centering\arraybackslash}p{.1cm}>{\centering\arraybackslash}p{1.4cm}>{\centering\arraybackslash}p{1.4cm}>{\centering\arraybackslash}p{.1cm}>{\centering\arraybackslash}p{1.4cm}>{\centering\arraybackslash}p{1.4cm}}
  \toprule
  & \multicolumn{2}{c}{\textbf{Soft}} & & \multicolumn{2}{c}{\textbf{Reserve (\textit{r}=1)}} & & \multicolumn{2}{c}{\textbf{Reserve (\textit{r}=1.05)}} \\
  \cline{2-3}\cline{5-6}\cline{8-9}
  $n$ & {prob. failure} & {expect. win. bid } &  & {prob. failure} & {expect. win. bid } & & {prob. failure} & {expect. win. bid }  \\ \midrule
 9 & 0.292 & 0.959 &  & 0.307 & 0.926 &  & 0.233 & 0.947 \\
  10 & 0.269 & 0.949 &  & 0.281 & 0.919 &  & 0.212 & 0.939 \\
  12 & 0.305 & 0.939 &  & 0.286 & 0.916 &  & 0.216 & 0.935 \\
  13 & 0.066 & 0.890 &  & 0.106 & 0.866 &  & 0.074 & 0.878 \\
  14 & 0.047 & 0.874 &  & 0.083 & 0.853 &  & 0.056 & 0.863 \\\bottomrule
\end{tabular}
\end{table}

\begin{figure}
  \centering
  \begin{subfigure}{0.51\textwidth}
    \includegraphics[width=\linewidth]{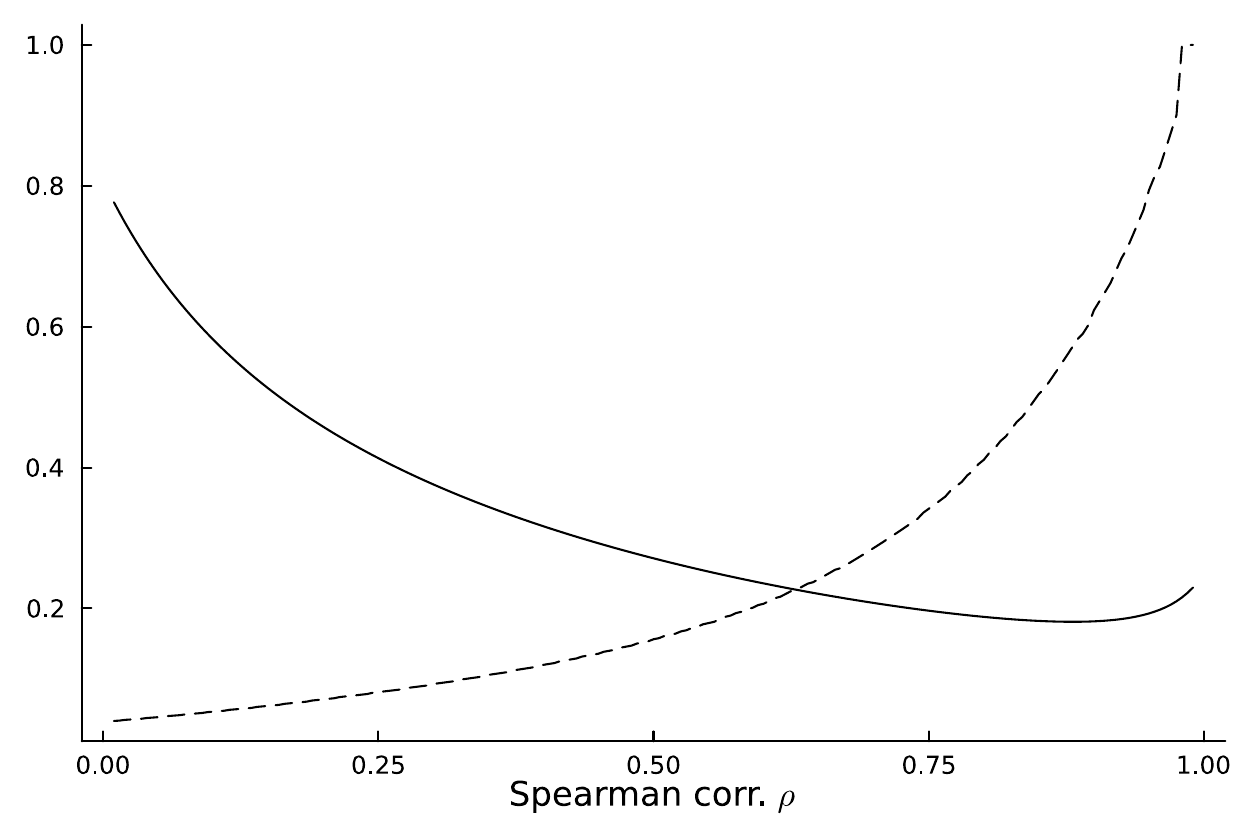}
    \caption{Auction failure probability \label{fig:counter_soft_rho_failure} }
  \end{subfigure}
  \vspace{1em}
  \begin{subfigure}{0.51\textwidth}
    \includegraphics[width=\linewidth]{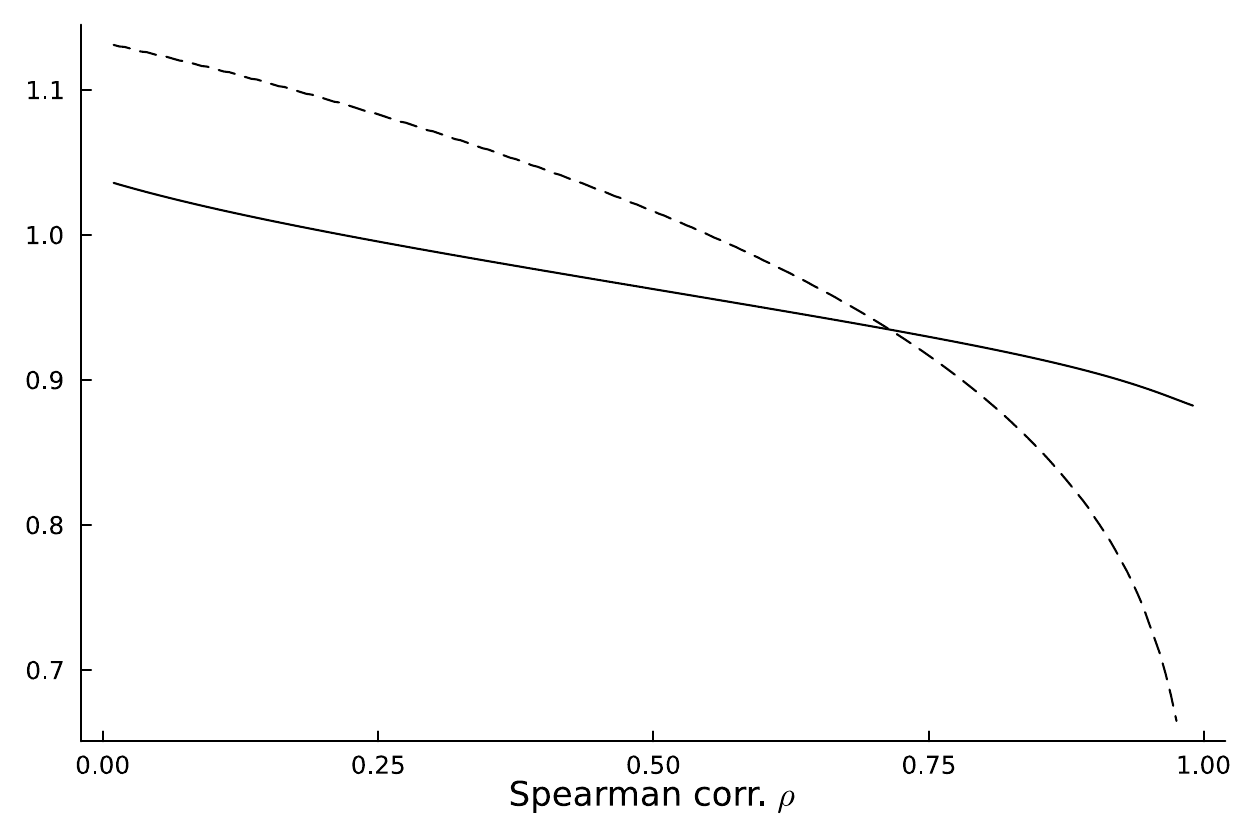}
    \caption{Expected winning bid \label{fig:counter_soft_rho_bid}}
  \end{subfigure}
\caption{Counterfactual auction failure probabilities and expected winning bids for different values of Spearman correlation $\rho$ under the soft bid requirement (dashed line) and reserve price ($r=1.05$, solid line) formats in auctions with $n=10$ potential bidders \label{fig:counter_soft_rho}}
\end{figure}

When the reserve price is $r=1$ (the engineer's estimate), the soft format produces lower auction failure probabilities for all $n$ except $n=12$. In particular, the soft format has lower auction failure probabilities by $1.5$, $1.2$, $4.0$, and $3.6$ percentage points for $n=9$, $10$, $13$, and $14$, respectively. However, the expected winning bid is higher under the soft format than under the reserve price format by $3.3\%$, $3.0\%$, $2.3\%$, $2.4\%$, and $2.1\%$ of the engineer's estimate for $n=9$, $10$, $12$, $13$, and $14$, respectively. Therefore, the auctioneer can relax the bid requirement constraint to reduce the probability of auction failure, but at the cost of a higher expected winning bid. The last two columns of Table \ref{tab:counterfact_soft} report the results for the reserve price format with $r=1.05$. In this case, the reserve price format achieves lower failure probabilities for $n=9$, $10$, and $12$ by $5.9$, $5.7$, and $8.9$ percentage points, while still maintaining a lower expected winning bid for all $n$. For $n=13$ and $14$, the soft format continues to have lower failure probabilities, but the differences are reduced to $0.8$ and $0.9$ percentage points, respectively. The expected winning bid under the reserve price format with $r=1.05$ remains lower than under the soft format by $1.2\%$ and $1.1\%$ of the engineer's estimate for $n=13$ and $14$, respectively.

Lastly, we compare the soft and reserve price formats for different levels of signal informativeness as measured by Spearman's rank correlation $\rho$. We consider auctions with $n=10$ potential bidders and set the entry cost to its estimated level for $n=10$: $5.2\%$ of the engineer's estimate. The reserve price is set to $r=1.05$. The results are reported in Figure \ref{fig:counter_soft_rho}. While the soft format can achieve a lower expected winning bid at high levels of signal informativeness, its auction failure probability is also very high and approaches one as $\rho$ increases. In particular, the reserve price format produces lower auction failure probabilities for $\rho>0.63$. Thus, a similar conclusion holds as for the hard bid requirement format: the reserve price format is preferred when signals are sufficiently informative, especially with respect to auction failure probability. In our application, the estimated level of signal informativeness is $\widehat \rho=0.68$, which falls in the range where the reserve price format dominates the soft bid requirement format in terms of both auction failure probability and expected winning bid.

\section{Conclusion}


An important outstanding question is what happens after an auction fails. The options include direct negotiation with suppliers, relisting, or cancellation as assumed in this paper.
In practice, the scope for direct negotiation can be limited: certain U.S.\ regulations, such as Title 23 of the Code of Federal Regulations governing federal-aid highway programs, prohibit direct negotiations with individual suppliers after an auction fails. Regarding the relisting option, in U.S.\ highway procurement, the auctioneer typically has strong incentives to modify the specifications of the original project to encourage greater competition. The modified project is then re-advertised in a new letting process, often after substantial delays. This effectively transforms the original project into a different one, and bidders would perceive the new auction for the relisted project as a distinct auction rather than a continuation of the failed auction.

In other institutional settings, the auctioned object may be carried over to subsequent rounds without modification, and bidders may strategically delay participation in anticipation of more favorable terms \citep[see][]{mcafee1997sequentially, liu2025dynamic}.\footnote{\cite{liu2025dynamic} study judicial sales auctions using a different model but demonstrate that the option of relisting can cause strategic delays, with bidders waiting for potentially lower reserve prices in subsequent rounds.} When the auctioneer is permitted to actively solicit more bids, in the subsequent rounds, there could be a larger pool of potential bidders. While our paper does not model these dynamics and proceeds under the simplifying assumption of auction cancellation, our analysis provides key insights into how bid requirements and reserve prices affect entry and procurement costs, and highlights the critical role of signal informativeness.
However, incorporating post-failure dynamics into an endogenous entry framework remains an important and challenging direction for future research.

Beyond post-failure dynamics, collusion may interact with the choice of auction format in important ways \citep[see][]{hendricks1989collusion}. Under bid requirements, cartel members can submit phantom bids to create the appearance of competition \citep{porter1993detection}, making auction failure unlikely unless signals are sufficiently unfavorable for all firms. Incorporating cartel structure into our framework would introduce asymmetries in both information (through information sharing among cartel members) and incentives (as cartel firms face no risk of cancellation due to insufficient competition). Whether tests based on departures from bid independence and exchangeability \citep{bajari2003collusion} can detect collusion in the presence of these asymmetries remains an open question.

\appendix
\section{Proofs of the main results}
\begin{proof}[Proof of Proposition \ref{p:entry}]
    By the copula properties, $F_{V\mid S}(v\mid s)=C_2(F(v),s)$.
    The expected revenue from entry for a bidder with signal $S=s$ when the entry probability is $p$ is given by
    \begin{align*}
      &  \int_{\underline v}^{\overline{v}} (\beta(v\mid p,n)-v)H(v\mid p,n)dC_2(F(v),s)\\
      &\quad= \int_{\underline v}^{\overline{v}}\left(\int_v^{\overline{v}} H(u\mid p,n)du\right) dC_2(F(v),s)\\
      &\quad= \int_{\underline v}^{\overline{v}} C_2(F(v),s) H(v\mid p,n) dv,
    \end{align*}
    where the equality in the second line holds by \eqref{eq:bidding}, and the equality in the last line holds by integration by parts and because $H(\overline{v}\mid p, n)=0$. The result follows by setting $s=p=p(n,\kappa)$ and applying the equilibrium condition in \eqref{eq:equilibrium condition}.
  \end{proof}

\begin{proof}[Proof of Proposition \ref{p:quasi-concavity}]
Denote the expected revenue of the marginal entrant as
\[R(p,n)\coloneqq \int_{\underline{v}}^{\overline{v}}C_{2}\left(F\left(v\right),p\right)\left\{ \left(1-C\left(F\left(v\right),p\right)\right)^{n-1}-\left(1-p\right)^{n-1}\right\} dv.\]
The derivative of $R(p,n)$ with respect to $p$ is
\[R_1(p,n)\coloneqq \frac{\partial R(p,n)}{\partial p}=(1-p)^{n-2}A(p,n)-B(p,n),\]
where
\begin{align*}
  A(p,n) &\coloneqq \left(n-1\right)\int_{\underline{v}}^{\overline{v}}C_{2}\left(F\left(v\right),p\right)dv-\left(1-p\right)\int_{\underline{v}}^{\overline{v}}C_{22}\left(F\left(v\right),p\right)dv\\
B\left(p,n\right) & \coloneqq  \left(n-1\right)\int_{\underline{v}}^{\overline{v}}C_{2}^{2}\left(F\left(v\right),p\right)\left(1-C\left(F\left(v\right),p\right)\right)^{n-2}dv\\
 &  \qquad -\int_{\underline{v}}^{\overline{v}}C_{22}\left(F\left(v\right),p\right)\left(1-C\left(F\left(v\right),p\right)\right)^{n-1}dv,
\end{align*}
and $A(p,n),B(p,n)>0$ since $C_{22}(\cdot,\cdot)<0$ by Assumption \ref{a:copula_1}\ref{a:goodnews}.

We show that $\sup_{p\in\left[\varepsilon,1\right]}R_{1}\left(p,n\right)<0$,
if $n$ is sufficiently large. First, since $C_{22}\left(\cdot,\cdot\right)<0$,
we have
\[
R_{1}\left(p,n\right)\leq\left(1-p\right)^{n-2}A\left(p,n\right)-\left(n-1\right)\int_{\underline{v}}^{\overline{v}}C_{2}^{2}\left(F\left(v\right),p\right)\left(1-C\left(F\left(v\right),p\right)\right)^{n-2}dv.
\]
It follows that
\[
\underset{p\in\left[\varepsilon,1\right]}{\sup}\left(1-p\right)^{n-2}A\left(p,n\right)\leq\left(1-\varepsilon\right)^{n-2}\bar{A}\left(n\right),
\]
where
\[
\bar{A}\left(n\right)\coloneqq\left(n-1\right)\int_{\underline{v}}^{\overline{v}}C_{2}\left(F\left(v\right),\varepsilon\right)dv+\underset{p\in\left[\varepsilon,1\right]}{\sup}\left|\int_{\underline{v}}^{\overline{v}}C_{22}\left(F\left(v\right),p\right)dv\right|.
\]
By these results,
\begin{equation}
R_{1}\left(p,n\right)\leq\left(1-\varepsilon\right)^{n-2}\bar{A}\left(n\right)-\left(n-1\right)\int_{\underline{v}}^{\overline{v}}C_{2}^{2}\left(F\left(v\right),p\right)\left(1-C\left(F\left(v\right),p\right)\right)^{n-2}dv.\label{eq:R_1_upper_bound}
\end{equation}
The copula density function is given by $C_{21}(x,y) \coloneqq \partial^2C(x,y)/\partial y \partial x$. Let $C_1(x,y) \coloneqq \partial C(x,y)/\partial x$. We show in the online Supplement that
\[
\int_{\underline{v}}^{\overline{v}}C_{2}^{2}\left(F\left(v\right),p\right)\left(1-C\left(F\left(v\right),p\right)\right)^{n-2}dv=\frac{1}{(n-2)^{2}}\left(\frac{2C_{21}^{2}(0,p)}{f(\underline{v})C_{1}^{3}(0,p)}+\delta_{n}\left(p\right)\right),
\]
where $\delta_{n}\left(p\right)\downarrow0$ uniformly in $p\in\left[\varepsilon,1\right]$.
By this result and \eqref{eq:R_1_upper_bound},
\[
\underset{p\in\left[\varepsilon,1\right]}{\sup}R_{1}\left(p,n\right)\leq-\frac{\left(n-1\right)}{(n-2)^{2}}\left(\underset{p\in\left[\varepsilon,1\right]}{\inf}\frac{2C_{21}^{2}(0,p)}{f(\underline{v})C_{1}^{3}(0,p)}-\underset{p\in\left[\varepsilon,1\right]}{\sup}\left|\delta_{n}\left(p\right)\right|-\frac{\left(n-2\right)^{2}}{\left(n-1\right)}\left(1-\varepsilon\right)^{n-2}\bar{A}\left(n\right)\right).
\]
Moreover, $(n-2)^{2}\left(1-\varepsilon\right)^{n-2}\bar{A}\left(n\right)/\left(n-1\right)\downarrow0$
as $n\uparrow\infty$. The conclusion follows from these results and
the fact that $\inf_{p\in\left[\varepsilon,1\right]}C_{21}^{2}(0,p)/C_{1}^{3}(0,p)>0$.
\end{proof}

\begin{proof}[Proof of Proposition \ref{p:cost_conditional}]
Suppose there are $N^*\geq 2$ active bidders. The CDF of the minimum value among the $N^*$ active bidders is 
\begin{equation*}
    1-(1-F^*(v\mid p))^{N^*},
\end{equation*}
and the corresponding expected winning bid when there are $N^*$ active bidders is
\begin{equation*}
    \int_{\underline v}^{\overline{v}}\beta(v\mid p,n) d(1-(1-F^*(v\mid p))^{N^*})=N^*\int_{\underline v}^{\overline{v}}\beta(v\mid p,n) (1-F^*(v\mid p))^{N^*-1} dF^*(v\mid p).
\end{equation*}
Therefore, conditional on at least two active bidders, the expected winning bid is
\begin{align}
    \lefteqn{\frac{1}{\Pr[N^*\geq 2\mid p,n]}
    \int_{\underline{v}}^{\overline{v}} \beta(v\mid p,n) \sum_{j=2}^n \binom{n}{j} jp^j(1-p)^{n-j} (1-F^*(v\mid p))^{j-1}dF^*(v\mid p)}\notag\\
    &\quad= \frac{n}{\Pr[N^*\geq 2\mid p,n]}
    \int_{\underline{v}}^{\overline{v}} \beta(v\mid p,n) \sum_{j=2}^n\binom{n-1}{j-1}p^j(1-p)^{n-j} (1-F^*(v\mid p))^{j-1}dF^*(v\mid p)\notag\\
    &\quad= \frac{np}{\Pr[N^*\geq 2\mid p,n]}
    \int_{\underline{v}}^{\overline{v}} \beta(v\mid p,n) \sum_{j=1}^{n-1}\binom{n-1}{j}p^j(1-p)^{n-1-j} (1-F^*(v\mid p))^{j}dF^*(v\mid p)\notag\\
    &\quad=  \frac{np}{\Pr[N^*\geq 2\mid p,n]}\int_{\underline{v}}^{\overline{v}} \beta(v\mid p,n) H(v\mid p,n)dF^*(v\mid p) \notag\\
    &\quad=-\frac{n}{\Pr[N^*\geq 2\mid p,n]}\int_{\underline{v}}^{\overline{v}} \beta(v\mid p,n) H(v\mid p,n)d\Lambda(v\mid p),\label{eq:PC1}
\end{align}
where the first equality holds by the binomial property $j\binom{n}{j} = n\binom{n-1}{j-1}$, the equality in the third line holds by the binomial theorem $\sum_{j=1}^{n-1}\binom{n-1}{j}(p(1-F^*(v\mid p)))^{j}(1-p)^{n-1-j}=(1-pF^*(v\mid p))^{n-1} -(1-p)^{n-1} $ and because $1-pF^*(v\mid p)=\Lambda(v\mid p)$. Applying integration by parts to the integral in \eqref{eq:PC1}, we obtain
\begin{align}
    \lefteqn{\int_{\underline{v}}^{\overline{v}} \beta(v\mid p,n) H(v\mid p,n)d\Lambda(v\mid p)} \notag\\
    &\quad =\beta(v\mid p,n) H(v\mid p,n) \Lambda(v\mid p)\Bigg|_{\underline{v}}^{\overline{v}}-\int_{\underline{v}}^{\overline{v}} \Lambda(v\mid p)d(\beta(v\mid p,n)H(v\mid p,n))\notag\\
    &\quad=-{\underline v}H(\underline v\mid p,n) -\int_{\underline{v}}^{\overline{v}} H(v\mid p,n) dv -\int_{\underline{v}}^{\overline{v}} \Lambda(v\mid p)d(\beta(v\mid p,n)H(v\mid p,n))\notag\\
    &\quad= -{\underline v}H(\underline v\mid p,n) +(1-p)^{n-1}(\overline{v} - \underline v) - \int_{\underline{v}}^{\overline{v}} \Lambda^{n-1}(v\mid p)  dv \notag \\
    &\quad\quad- \int_{\underline{v}}^{\overline{v}} \Lambda(v\mid p)d(\beta(v\mid p,n)H(v\mid p,n))\notag\\
    &\quad=\overline{v} (1-p)^{n-1} -{\underline v}- \int_{\underline{v}}^{\overline{v}} \Lambda^{n-1}(v\mid p)  dv - \int_{\underline{v}}^{\overline{v}} \Lambda(v\mid p)d(\beta(v\mid p,n)H(v\mid p,n)),\label{eq:PC2}
\end{align}
where the second equality holds by $H(\overline{v}\mid p,n)=0$ and the equilibrium bidding strategy in \eqref{eq:bidding}, and the last equality holds by $\Lambda(\underline v\mid p)=1$. By the first-order condition for the equilibrium bidding strategy,
\begin{align*}
    d(\beta(v\mid p,n)\cdot H(v\mid p,n))/ dv=&v\cdot H'(v\mid p,n),
\end{align*}
and the second integral in \eqref{eq:PC2} becomes 
\begin{align}
   \lefteqn{ (n-1)\int_{\underline{v}}^{\overline{v}} \Lambda^{n-1}(v\mid p)vd\Lambda(v\mid p)} \notag\\
   &\quad=  \frac{n-1}{n}\int_{\underline{v}}^{\overline{v}} vd\Lambda^n(v\mid p) \notag\\
   &\quad = \frac{n-1}{n}(\overline{v} (1-p)^n - \underline{v})-\frac{n-1}{n}\int_{\underline{v}}^{\overline{v}} \Lambda^n(v\mid p) dv. \label{eq:PC3}
\end{align}
Combining \eqref{eq:PC2} and \eqref{eq:PC3} and multiplying by $n$, we obtain
\begin{align*}
\lefteqn{ n\int_{\underline{v}}^{\overline{v}} \beta(v\mid p,n) H(v\mid p,n)d\Lambda(v\mid p)}\notag \\
  &\quad= \overline{v} ((1-p)^n +np(1-p)^{n-1}) -{\underline v} 
  -  n\int_{\underline{v}}^{\overline{v}} \Lambda^{n-1}(v\mid p)\left(1-\frac{n-1}{n}\Lambda(v\mid p)\right)dv.
\end{align*}
\end{proof}

\begin{proof}[Proof of equation \eqref{eq:win_Bid_reserve}]
  When the entry probability is $p$, the CDF of private costs of an active bidder (i.e., conditional on $S\leq p$ and $V\leq r$) is 
  \[F^*(v\mid p,r)\coloneqq C(F(v),p)/C(F(r),p).\] The CDF of the minimum private cost among $N^*$ active bidders is $1-(1-F^*(v\mid p,r))^{N^*}$. As in the proof of Proposition \ref{p:cost_conditional}, the expected winning bid is
  \begin{align*}
    &\frac{1}{\Pr[N^*\geq 1\mid p,n,r]}\int_{\underline v}^{r} \beta(v\mid p,n,r) \Bigg(\sum_{j=1}^n \binom{n}{j} jC^j(F(r),p)(1-C(F(r),p))^{n-j} \\
    &\qquad\qquad\qquad\qquad\qquad\qquad\qquad\qquad \times (1-F^*(v\mid p,r))^{j-1}\Bigg) dF^*(v\mid p,r).
  \end{align*}
  The rest of the proof follows the same steps as in the proof of Proposition \ref{p:cost_conditional}.
  \end{proof}

  \begin{proof}[Proof of Proposition \ref{prop:limit_theta_infty}] We assume $\lim_{\theta\uparrow\infty}C(x,y;\theta)=\min\{x,y\}$. Therefore, $\lim_{\theta\uparrow\infty}C_2(x,y;\theta)=\mathbbm{1}(x\geq y)$. For part \ref{prop:limit_theta_bid_req}, we have 
\begin{align*}
   \lim_{\theta\uparrow\infty} \Pi(p,n,\kappa,p;\theta) 
   & = \int_{\underline v}^{\overline{v}} \mathbbm{1}(F(v)\geq p)\left((1-\min\{F(v),p\})^{n-1}-(1-p)^{n-1}\right)dv
   -\kappa
    =-\kappa,
\end{align*}
where the first equality holds by \eqref{eq:Pi_bid_requirement}.
For part \ref{prop:limit_theta_res_price}, suppose $F^{-1}(p)\leq r$. By \eqref{eq:Pi_reserve_price},
\begin{align*}
     \lim_{\theta\uparrow\infty} \Pi(p,n,\kappa,r,p;\theta)
     &=\int_{\underline v}^r \mathbbm{1}(F(v)\geq p)(1-\min\{F(v),p\})^{n-1}dv
   -\kappa\\
   &=\int_{F^{-1}(p)}^r(1-p)^{n-1}dv-\kappa \\
   &= (r-F^{-1}(p))(1-p)^{n-1}-\kappa.
\end{align*}
If $F^{-1}(p)\geq r$, the expected revenue is zero in the limit.
    
\end{proof}

  \begin{proof}[Proof of equation \eqref{eq:p_N_ID}] Consider an auction with $N_l=n\geq 2$ potential bidders. We have
    \begin{align*}
      &\Pr[S_{1l}\leq p_n\mid \bar{N}_{l}\geq 2, N_l=n] \\
      &\quad=
      \frac{\Pr[S_{1l}\leq p_n\,\text{and } S_{jl}\leq p_n \text{ for some } j=2,\ldots,n \mid N_l=n]}{\Pr[ \bar{N}_{l}\geq 2 \mid N_l=n]}\\
      &\quad =\frac{\Pr[S_{1l}\leq p_n ] \bigg(1- \Pr[ S_{jl}> p_n \text{ for all } j=2,\ldots,n \mid N_l=n]\bigg)}{1-\Pr[ \bar{N}_{l}<2 \mid N_l=n]},
    \end{align*}
  where first equality holds by Assumption \ref{a:ID1}\ref{a:independent}.
  \end{proof}


\bibliographystyle{econometrica}
\bibliography{entry_cost,entry_econometrics}

@book{FTA2016_Best,
    author = {{Federal Transit Administration}},
    title = {Best Practices Procurement \& Lessons Learned Manual. {FTA} Report No. 0105},
    publisher = {U.S. Department of Transportation} ,
    year = {2016},
    month = {October}
}

@article{bajari2003collusion,
  title     = {Deciding Between Competition and Collusion},
  author    = {Bajari, Patrick and Ye, Lixin},
  journal   = {Review of Economics and Statistics},
  volume    = {85},
  number    = {4},
  pages     = {971--989},
  year      = {2003}
}

@article{bhattacharya2014regulating,
  title     = {Regulating Bidder Participation in Auctions},
  author    = {Bhattacharya, Vivek and Roberts, James W and Sweeting, Andrew},
  journal   = {RAND Journal of Economics},
  volume    = {45},
  number    = {4},
  pages     = {675--704},
  year      = {2014},
  publisher = {Wiley Online Library}
}

@article{chen2024identification,
  title   = {{Identification and Inference in First-Price Auctions with Risk Averse Bidders and Selective Entry}},
  author  = {Chen, Xiaohong and Gentry, Matthew L and Li, Tong and Lu, Jingfeng},
  journal = {Review of Economic Studies},
  year    = {2026},
  volume  = {93},
  number  = {1},
  pages   = {366--403},
  doi     = {10.1093/restud/rdaf016}
}

@article{fang2014inference,
  title     = {Inference of Bidders’ Risk Attitudes in Ascending Auctions with Endogenous Entry},
  author    = {Fang, Hanming and Tang, Xun},
  journal   = {Journal of Econometrics},
  volume    = {180},
  number    = {2},
  pages     = {198--216},
  year      = {2014},
  publisher = {Elsevier}
}

@article{GL,
  title     = {Identification in Auctions with Selective Entry},
  author    = {Gentry, Matthew and Li, Tong},
  journal   = {Econometrica},
  volume    = {82},
  number    = {1},
  pages     = {315--344},
  year      = {2014},
  publisher = {Wiley Online Library}
}

@article{GPV,
  title     = {Optimal Nonparametric Estimation of First-Price Auctions},
  author    = {Guerre, Emmanuel and Perrigne, Isabelle and Vuong, Quang},
  journal   = {Econometrica},
  volume    = {68},
  number    = {3},
  pages     = {525--574},
  year      = {2000},
  publisher = {Wiley Online Library}
}

@article{hendricks1989collusion,
  title     = {Collusion in Auctions},
  author    = {Hendricks, Kenneth and Porter, Robert H},
  journal   = {Annales d'{\'E}conomie et de Statistique},
  number    = {15/16},
  pages     = {217--230},
  year      = {1989}
}

@article{hong2002increasing,
  title     = {Increasing Competition and the Winner's Curse: Evidence from Procurement},
  author    = {Hong, Han and Shum, Matthew},
  journal   = {Review of Economic Studies},
  volume    = {69},
  number    = {4},
  pages     = {871--898},
  year      = {2002},
  publisher = {Wiley-Blackwell}
}

@article{kang2022winning,
  title     = {{Winning by Default: Why Is There So Little Competition in Government Procurement?}},
  author    = {Kang, Karam and Miller, Robert A},
  journal   = {Review of Economic Studies},
  volume    = {89},
  number    = {3},
  pages     = {1495--1556},
  year      = {2022},
  publisher = {Oxford University Press}
}

@article{krasnokutskaya2011bid,
  title     = {Bid Preference Programs and Participation in Highway Procurement Auctions},
  author    = {Krasnokutskaya, Elena and Seim, Katja},
  journal   = {American Economic Review},
  volume    = {101},
  number    = {6},
  pages     = {2653--2686},
  year      = {2011},
  publisher = {American Economic Association}
}

@article{krasnokutskaya2011identification,
  title     = {Identification and Estimation of Auction Models with Unobserved Heterogeneity},
  author    = {Krasnokutskaya, Elena},
  journal   = {Review of Economic Studies},
  volume    = {78},
  number    = {1},
  pages     = {293--327},
  year      = {2011},
  publisher = {Oxford University Press}
}

@book{krishna2010auction,
  title     = {Auction Theory},
  author    = {Krishna, Vijay},
  year      = {2010},
  edition   = {Second},
  publisher = {Academic Press}
}

@article{levin1994equilibrium,
  title     = {{Equilibrium in Auctions with Entry}},
  author    = {Levin, Dan and Smith, James L},
  journal   = {American Economic Review},
  volume    = {84},
  number    = {3},
  pages     = {585--599},
  year      = {1994}
}

@article{li2009entry,
  title     = {{Entry and Competition Effects in First-Price Auctions: Theory and Evidence from Procurement Auctions}},
  author    = {Li, Tong and Zheng, Xiaoyong},
  journal   = {Review of Economic Studies},
  volume    = {76},
  number    = {4},
  pages     = {1397--1429},
  year      = {2009},
  publisher = {Wiley-Blackwell}
}

@article{mcafee1997sequentially,
  title     = {Sequentially Optimal Auctions},
  author    = {McAfee, R Preston and Vincent, Daniel},
  journal   = {Games and Economic Behavior},
  volume    = {18},
  number    = {2},
  pages     = {246--276},
  year      = {1997}
}

@article{MSX,
  title     = {What Model for Entry in First-Price Auctions? A Nonparametric Approach},
  author    = {Marmer, Vadim and Shneyerov, Artyom and Xu, Pai},
  journal   = {Journal of Econometrics},
  volume    = {176},
  number    = {1},
  pages     = {46--58},
  year      = {2013},
  publisher = {Elsevier}
}

@book{nelsen2007introduction,
  title     = {An Introduction to Copulas},
  author    = {Nelsen, Roger B},
  year      = {2006},
  publisher = {Springer}
}

@article{quint2018theory,
  title     = {A Theory of Indicative Bidding},
  author    = {Quint, Daniel and Hendricks, Kenneth},
  journal   = {American Economic Journal: Microeconomics},
  volume    = {10},
  number    = {2},
  pages     = {118--151},
  year      = {2018},
  publisher = {American Economic Association}
}

@article{porter1993detection,
  title     = {Detection of Bid Rigging in Procurement Auctions},
  author    = {Porter, Robert H and Zona, J Douglas},
  journal   = {Journal of Political Economy},
  volume    = {101},
  number    = {3},
  pages     = {518--538},
  year      = {1993}
}

@article{roberts2013should,
  title     = {When Should Sellers Use Auctions?},
  author    = {Roberts, James W and Sweeting, Andrew},
  journal   = {American Economic Review},
  volume    = {103},
  number    = {5},
  pages     = {1830--1861},
  year      = {2013},
  publisher = {American Economic Association}
}

@article{samuelson1985competitive,
  title     = {Competitive Bidding with Entry Costs},
  author    = {Samuelson, William F},
  journal   = {Economics Letters},
  volume    = {17},
  number    = {1-2},
  pages     = {53--57},
  year      = {1985},
  publisher = {Elsevier}
}

@unpublished{titl2023one,
  title  = {The One and Only: Single Bidding in Public Procurement},
  author = {Titl, V{\'\i}t{\v{e}}zslav},
  note   = {SSRN Working Paper 3954295},
  year   = {2023}
}

@misc{TxDOT2014,
  title  = {Standard Specifications for Construction and Maintenance of Highways, Streets, and Bridges},
  author = {{Texas Department of Transportation}},
  year   = {2014},
  month  = {11},
  note   = {Adopted November 1, 2014}
}

@article{xu2013nonparametric,
  title     = {Nonparametric Estimation of Entry Cost in First-Price Procurement Auctions},
  author    = {Xu, Pai},
  journal   = {Journal of Applied Econometrics},
  volume    = {28},
  number    = {6},
  pages     = {1046--1065},
  year      = {2013},
  publisher = {Wiley Online Library}
}

@unpublished{liu2025dynamic,
  title  = {Dynamic Bidding under Relisting and New Entry: Evidence from Judicial Sales},
  author = {Liu, Hui and Luo, Yao and Xiao, Ruli},
  year   = {2025},
  note   = {Working paper}
}

@article{ye2007indicative,
  title     = {Indicative Bidding and a Theory of Two-Stage Auctions},
  author    = {Ye, Lixin},
  journal   = {Games and Economic Behavior},
  volume    = {58},
  number    = {1},
  pages     = {181--207},
  year      = {2007},
  publisher = {Elsevier}
}

@Article{Hickman:2014kq,
  author   = {Hickman, Brent R. and Hubbard, Timothy P.},
  journal  = {Journal of Applied Econometrics},
  title    = {{Replacing Sample Trimming with Boundary Correction in Nonparametric Estimation of First-Price Auctions}},
  year     = {2015},
  issn     = {1099-1255},
  number   = {5},
  pages    = {739--762},
  volume   = {30},
  doi      = {10.1002/jae.2385},
}

@Article{Zincenko2024,
  author    = {Zincenko, Federico},
  journal   = {Journal of Econometrics},
  title     = {{Estimation and Inference of Seller's Expected Revenue in First-Price Auctions}},
  year      = {2024},
  issn      = {0304-4076},
  number    = {1},
  pages     = {105734},
  volume    = {241},
  doi       = {10.1016/j.jeconom.2024.105734},
}

@Article{Ma2021,
  author    = {Ma, Jun and Marmer, Vadim and Shneyerov, Artyom and Xu, Pai},
  journal   = {Econometric Reviews},
  title     = {{Monotonicity-Constrained Nonparametric Estimation and Inference for First-Price Auctions}},
  year      = {2021},
  issn      = {0747-4938},
  number    = {10},
  pages     = {944--982},
  volume    = {40},
  doi       = {10.1080/07474938.2021.1889198},
}

@Article{Ma2019,
  author    = {Ma, Jun and Marmer, Vadim and Shneyerov, Artyom},
  journal   = {Journal of Econometrics},
  title     = {{Inference for First-Price Auctions with Guerre, Perrigne, and Vuong's Estimator}},
  year      = {2019},
  issn      = {0304-4076},
  number    = {2},
  pages     = {507--538},
  volume    = {211},
  doi       = {10.1016/j.jeconom.2019.02.006},
  eprint    = {1903.06401},
}
\end{document}